\documentclass[journal]{IEEEtran}
  \usepackage{amsfonts}
  \usepackage{mathrsfs}
  \usepackage{amssymb,amsmath}
  \usepackage[compress]{cite}
  \usepackage{bm}
  \usepackage{algorithm}
  \usepackage{algorithmic}
  \usepackage{epsfig,graphics,subfigure}
  \usepackage{color}
  \usepackage{multirow}
  \usepackage{graphicx}
  \usepackage{amsthm}
  \usepackage{array}
  \usepackage{tabularx}
  \usepackage{multicol}
  \usepackage{booktabs}
  \usepackage{cite}
  \usepackage{stfloats}
  \usepackage{balance}
    \usepackage{epstopdf}

\newtheorem{lemma}{Lemma}

\newtheorem{proposition}{Theorem}

\begin{document}

\title{Optimal Beamforming and Time Allocation for Partially Wireless Powered Sensor Networks with Downlink SWIPT }
\author{Shiqi Gong, Shaodan Ma, Chengwen Xing and Guanghua Yang
%
}

\maketitle

\begin{abstract}
  Wireless powered sensor networks (WPSNs) have emerged as a key development towards the future self-sustainable Internet of Things (IoT) networks. To achieve a good balance between self-sustainability and reliability, partially WPSNs with a mixed power solution are desirable for practical applications. Specifically, most of the sensor nodes are wireless powered but the key sensor node adopts traditional wire/battery power for reliability. As a result, this paper mainly investigates optimal design for the partially WPSNs in which simultaneous wireless information and power transfer (SWIPT) is adopted in the downlink. Two scenarios with space division multiple access (SDMA) and  time division multiple access (TDMA) in the uplink are considered. For both the SDMA-enabled and TDMA-enabled partially WPSNs, joint design of downlink beamforming, uplink beamforming and time allocation is investigated to maximize the uplink sum rate while guaranteeing the quality-of-service (i.e., satisfying the downlink rate constraint) at the key sensor node. After analyzing the feasibility of uplink sum rate maximization problems and the influence of the downlink rate constraint, semi-closed-form optimal solutions for both SDMA-enabled and TDMA-enabled WPSNs are proposed with guaranteed global optimality. Complexity analysis is also provided to justify the advantage of the proposed solutions in low complexity. The effectiveness and optimality of the proposed optimal solutions are finally demonstrated by simulations.
 %
\end{abstract}

\section{Introduction}
As a new paradigm in communications, internet of things (IoT) can provide intelligent control and smart solutions for various tasks in our lives by connecting a large variety of devices \cite{R2,R0,R1}. It has been gradually applied in various applications from smart homes, healthcare to structural and environmental monitoring, and disaster warning and so on \cite{WSN0,WSN1,WSN2,WSN3,WSN4}. Usually IoT networks involve a large number of sensor nodes to collect and exchange information with data center and can also be regarded as wireless sensor networks (WSNs). The success of IoT networks heavily relies on the reliability and sustainability of the sensor nodes. When a large number of sensors are deployed, power supply for the sensor nodes becomes a challenging issue to be solved. Currently, there are several available solutions to power up the sensors. The first one is to wire the sensor to a fixed power supply through cables. The installation is time consuming and location dependent, and the wire connection also limits the mobility of sensors. The second one is to power sensors by batteries. However, batteries usually have short lifetime and their maintenance and replacement are costly and difficult, especially when sensors are deployed in harsh environment or remote locations. It is even impossible when the sensors  are deployed  inside the  building  structures or human bodies\cite{WSN0,WSN1}. The third one is to self-harvest energy from natural energy sources, such as solar and wind. But the amount of harvested energy is unstable and affected by uncontrollable nature factors. Lately, a new solution ``wireless power transfer (WPT)" was proposed \cite{WSN2, WSN3,WSN4}. It leverages the fact that energy could be transferred wirelessly through radio frequency (RF) signals. Compared to other natural based energy harvesting, the RF oriented energy harvesting is generally ubiquitous, predictable and steady  with  low cost\cite{R4, RR4}. As reported in \cite{R4}, the energy harvesters operating at 915MHz and using Dipole antennas can collect about 3.7mW and 1uW  of wireless power from  RF signals at  distances of 0.6m  and 6m, respectively. Meanwhile, advanced antenna and transceiver designs for realizing high RF energy harvesting efficiency have also been reported \cite{RR4}. With high feasibility and a wide range of applications in IoT, wireless powered sensor networks (WPSNs) have thus gained considerable research interest recently.

Generally in WPSNs, wireless powered sensors firstly harvest energy from the downlink RF signal transmitted by a power source or a hybrid access point (H-AP) which serves dually as a power source and a data center, and then utilize  the harvested  energy  for   uplink information  transmission \cite{R5, R6,RWPCN2}. With multiple sensors, the uplink transmission can be supported following spatial division multiple access (SDMA) or time division multiple access (TDMA) schemes. To achieve various objectives, optimal designs for SDMA-enabled and TDMA-enabled WPSNs are necessary and have been investigated in the literature \cite{ RWPCN2, RWPCN3,RWPCN30,RWPCN00,RWPCN0, RWPCN1, RWPCN5}. Specifically, for SDMA-enabled WPSNs, the optimal  H-AP
beamforming  was  proposed for maximizing the uplink sum rate and maximizing the minimum uplink rate among multiple wireless powered sensors, respectively in \cite{RWPCN2} and \cite{RWPCN3}. In \cite{RWPCN30}, uplink sum throughput maximization under various cooperation protocols in SDMA-enabled cognitive WPSNs was studied. Moreover, fairness-based  uplink  throughput  maximization was discussed for multiple-input-multiple-output (MIMO) WPSNs in \cite{RWPCN00}. With respect to the TDMA-enabled WPSNs, \cite{RWPCN0, RWPCN1} studied the optimal time allocation for multiple energy harvesters to maximize
  uplink sum throughput  for  single-input-single-ouput (SISO) WPSNs. When multiple-input-single-output (MISO) TDMA-enabled WPSNs were considered, joint beamforming design and time allocation for uplink  sum throughput maximization were investigated in \cite{RWPCN5}. It was found that the downlink  energy beamforming design for the H-AP was similar to that in the SDMA-enabled WPSNs. For TDMA-enabled WPSNs with separate multiple-antenna power source and single-antenna data center, sum throughput maximization through optimal beamforming and time allocation was investigated in \cite{WSN3}. Optimal solutions were proposed for two different scenarios, i.e.,
 the power source and the sensor nodes belong to the same or different service operator(s).

All the aforementioned optimal designs consider fully wireless powered sensor networks, in which all the sensor nodes are wireless powered and the downlink is dedicated for wireless power transfer only\footnote{If the H-AP has information to be transmitted to the sensor nodes, the information can be transmitted in another dedicated downlink phase.}. However in practice, mixed power solution may be adopted to achieve enhanced reliability in WSNs. Particularly, most of the sensor nodes are wireless powered, but the key sensor node is powered up by traditional battery/wire power for reliable communications. This network can be regarded as a partially wireless powered sensor network. In this partially WPSN, dual functions of RF signals in wireless information and power transfer can be exploited simultaneously in the downlink to further improve the spectral efficiency. In other words, the downlink RF signals can carry not only energy to wireless powered sensors but also information to the key sensor node with traditional power supply, which results in simultaneous wireless information and power transfer (SWIPT) in the downlink\cite{R40,R41}. This partially WPSN with downlink SWIPT can achieve efficient communications with a simple mixed power supply solution and thus is attractive for practical applications\cite{RWPCN6}. Nevertheless, optimal design for such WPSNs with downlink SWIPT is rarely investigated due to the difficulty in coupled downlink and uplink design as well as mixed power and information transfer. As far as we know, there is only one available optimal design for WPSNs with downlink SWIPT reported in \cite{RWPCN6}. It considered a SDMA-enabled WPSN with multiple users. Multiple downlink SWIPT phases were introduced to sequentially transmit information to one sensor while power up the others and equal time duration was assumed for all the downlink and uplink phases. Joint design of downlink beamforming and uplink power allocation were investigated to maximize the minimum downlink and uplink signal-to-interference-and-noise ratios (SINRs).

Here we take a step forward to investigate joint design of beamforming and time allocation for partially WPSNs with downlink SWIPT. To guarantee quality of service (QoS) at the key sensor node with traditional power supply, downlink rate constraint is taken into account and uplink sum rate maximization under such downlink rate constraint is mainly concerned. The implementation of the downlink SWIPT as well as the consideration of downlink rate constraint makes the optimal design challenging and also differentiates our design from the prior ones \cite{RWPCN2,RWPCN3,RWPCN30,RWPCN00,RWPCN0,RWPCN1, RWPCN5,RWPCN6}. Both SDMA and TDMA schemes are considered for the uplink transmission. Notice that extra time allocation in the uplink is involved in the TDMA-enabled WPSNs with downlink SWIPT and its optimal design hasn't been discussed before. The formulated uplink sum rate maximization problems for both SDMA-enabled and TDMA-enabled partially WPSNs are originally non-convex with coupled optimization variables. But they can be reformulated as concave problems through certain transformation. After analyzing the feasibility of uplink sum rate maximization problems and the influence of the downlink rate constraint, semi-closed-form optimal solutions for both SDMA-enabled and TDMA-enabled WPSNs are proposed with guaranteed global optimality. Complexity analysis is also provided to justify the advantage of our proposed solutions in low complexity. The effectiveness and optimality of our proposed optimal solutions are finally demonstrated by simulations.

The rest of the paper is organized as follows. The partially WPSNs are introduced and the optimal design problems are formulated in Section II. Feasibility analysis and problem reformulation are given in Section III. Semi-closed-form optimal solutions for SDMA-enabled and TDMA-enabled WPSNs are derived and meaningful insights in the optimal solutions are pointed out in Section IV. Simulation results are presented in Section V and finally conclusions are drawn in Section VI.

\noindent {\textit{\textbf{Notation}}}:
Throughout this paper, bold-faced lowercase and uppercase letters stand for vectors and matrices, respectively. The symbols ${\bm{A}}^T$, ${\bm{A}}^{\ast}$, ${\bm{A}}^H$, ${\bm{A}}^{-1}$ and $\text{Tr}(\bm{A})$ denote transpose, conjugate, Hermitian, inverse and trace of matrix $\bm{A}$, respectively. In addition, $\text{vec}(\cdot)$ denotes the vector formed by stacking the columns of a  matrix, while
$\text{diag}[\cdot]$  represents
a 
 diagonal matrix constructed from a vector. $\bm{I}_{d}$ denotes a $d$-dimensional identity matrix. All eigenvalues and singular values in our work  are arranged in a decreasing order.
  Finally, $a \rightarrow b$ indicates  $a$ approaches $b$  and  $(a)^{+}=\max(a, 0)$.

 \section{system model and problem formulation  }
 As shown in Fig. 1,  we consider a partially WPSN with $K$ wireless powered sensor nodes, also called energy harvesters (ERs), $E_k,  \forall k \!\!\in \!\!\mathcal{K}, \mathcal{K}\!\!=\!\!\{1,\cdots,K\}$, one battery/wire powered sensor node, also called  information receiver  (IR), and one hybrid access point (H-AP). The H-AP serves as not only power source to power up the sensor nodes ERs but also the data center to communicate with all sensor nodes. The H-AP and ERs are equipped with $N_B$ and $N_U$ antennas for effective wireless power transfer and harvesting, respectively, while the single-antenna IR is considered. This mixedly powered sensor network leverages wireless power transfer technique to solve for the challenging power supply problem for the majority of the sensor nodes, while adopts traditional battery/wire power only at the key sensor node to guarantee its communications. It can achieve efficient communications with simple power supply solution and has a wide range of practical applications.

In the partially WPSNs, two different phases are involved for power transfer and communications. In the first downlink phase, the H-AP transfers power to the ERs and sends information to the IR simultaneously. In other words, SWIPT is conducted in the downlink. Then in the second uplink phase, all the sensor nodes transmit sensing data to the H-AP. Since the ERs do not have a fixed power supply, their data transmissions are only powered by the harvested energy in the first downlink phase.

Defining the total time duration for the two phases as a unit, we assume the first $\tau_0$ $(0\le\tau_0\le 1)$ slot is utilized for downlink transmission and the remaining $(1-\tau_0)$ slot is allocated for uplink transmission. In the uplink transmission, two multiple access schemes, i.e., SDMA and TDMA, are considered. Specifically, for the SDMA-enabled WPSN, the IR and all ERs simultaneously transmit information to the H-AP in the $(1-\tau_0)$ slot, while for the TDMA-enabled WPSN, each ER and the IR sequentially transmit information to the H-AP in the $\tau_{E_k}, \forall k \!\!\in \!\!\mathcal{K}$  and $\tau_{I_R}$ slots, respectively, with $\sum\limits_{k=1}^{K}\tau_{E_k}+\tau_{I_R}
\!=\!1\!-\!\tau_0 $. In general, SDMA outperforms TDMA in terms of uplink sum rate. However, TDMA is easy to implement with low signal detection complexity at the receiver. Both of them are widely adopted in wireless communication systems \cite{ RWPCN5,RWPCN6}.
\begin{figure}[t]
\centering
\includegraphics[width=3.2in]{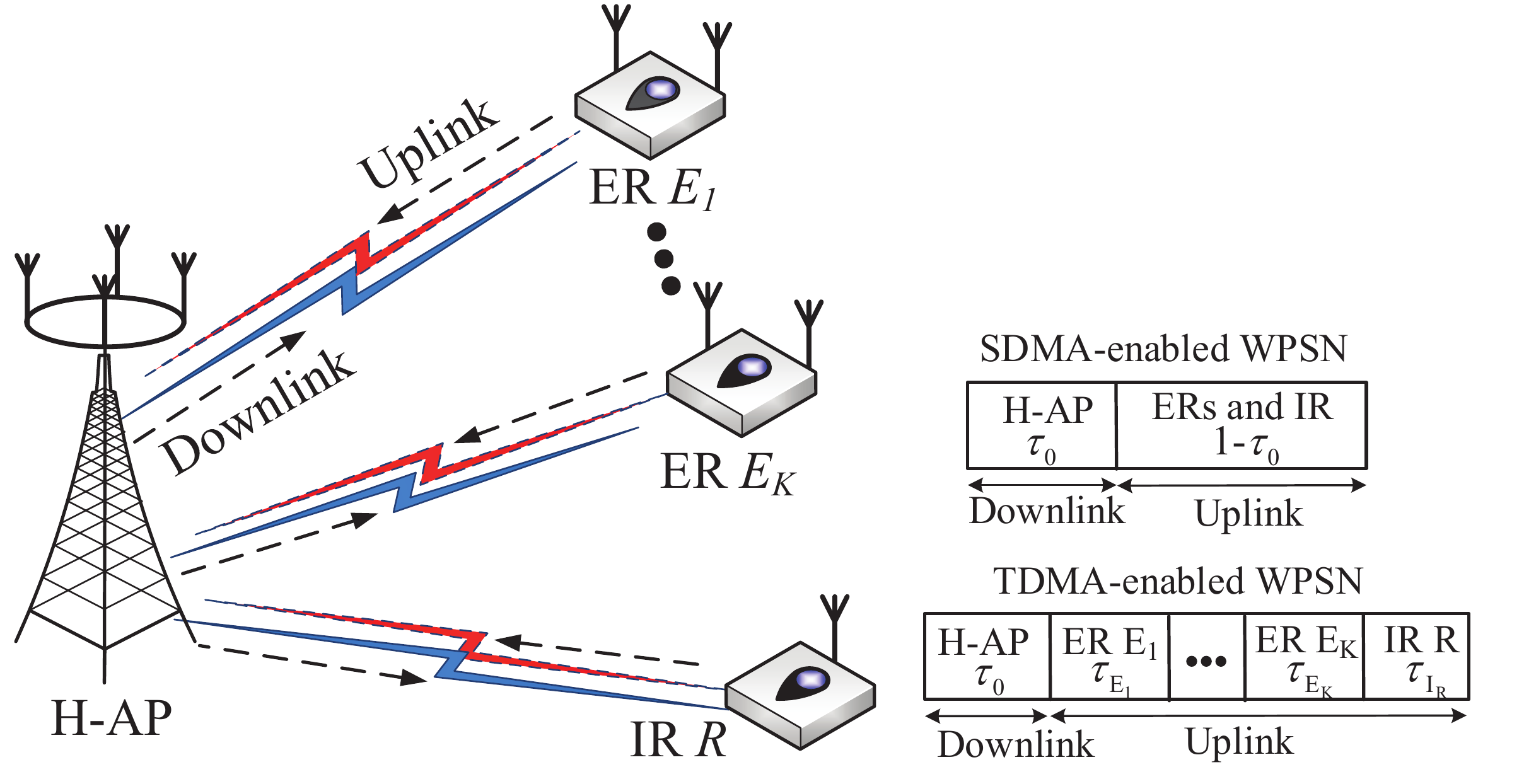}
\caption{SDMA-enabled and  TDMA-enabled MIMO WPSNs.}
\label{figure0}
\end{figure}

\subsection{SDMA-enabled WPSN}
In the first downlink phase, the H-AP adopts the SWIPT technique to transmit energy and information to  $K$ ERs and the IR simultaneously. The energy-carrying information signal is denoted as $\bm{s}_{B}\in \mathbb{C}^{N_B} $ with covariance matrix $\bm{W}_B=\mathbb{E}[\bm{s}_{B}\bm{s}_{B}^H]  \!\in\!\mathbb{C}^{N_B \times N_B}$. The transmission power can be written as $\text{tr}(\bm{W}_B)$  and should usually satisfy the power constraint $\text{tr}(\bm{W}_B) \le P_B$, where $P_B$ is the maximum allowable transmission power. Denoting the downlink channels from the H-AP
  to the IR   and  ER $E_k$ as
 $\bm{h}_{I_R}\in\mathbb{C}^{ N_B}$ and $\bm{H}_{E_k}\in\mathbb{C}^{N_U \times N_B}, \forall k \!\!\in \!\!\mathcal{K}$, respectively, the received signals at $K$ ERs and IR can be respectively written as
\begin{align}
 \bm{y}_{E_k} & =\bm{H}_{E_k}\bm{s}_{B}+\bm{n}_{E_k} , \forall k \!\in \!\mathcal{K} \label{eq1}\\
{y}_{I_R}&=\bm{h}_{I_R}^H
\bm{s}_B+{n}_{I_R}\label{eq2}
 \end{align} where $\bm{n}_{E_k}\sim\mathcal{CN}(\bm{0}, \sigma_n^2\bm{I}_{N_U})$ and ${n}_{I_R}\sim\mathcal{CN}({0}, \sigma_n^2)$ are the additive white Gaussian noises (AWGNs) at the ER $E_k, \forall k \!\!\in \!\!\mathcal{K}$ and IR, respectively. According to \eqref{eq1}, the   harvested energy at ER $E_k$ in this $\tau_0$ slot can be expressed as\vspace{-1mm}
 \begin{align}\label{eq02}
   Q_{E_k}=\tau_0\varepsilon_k\text{tr}(\bm{H}_{E_k}
 \bm{W}_B
 \bm{H}_{E_k}^H),~\forall k \!\in \!\mathcal{K}
 \end{align} where $0\le \varepsilon_k \le 1$  denotes  the energy  harvesting efficiency  of ER $E_k, ~\forall k \!\in \!\mathcal{K}$.
Meanwhile, based on \eqref{eq2}, the achievable downlink rate of the IR is expressed  as
\begin{align}
 R_{I_R}^D & \!=\!\tau_0\log(1\!+
 \!\sigma_n^{-2}\bm{h}_{I_R}^H\bm{W}_B
 \bm{h}_{I_R}).
\label{eq3}
 \end{align}

In the second uplink stage, the IR and ER $E_k$ simultaneously  transmit the information signals ${s}_{R}$ with $\mathbb{E}[\vert{s}_{R}\vert^2]\!\!=\!\!1$ and $\bm{x}_{E_k}\!\!\in\!\!\mathbb{C}^{N_U}$ with covariance matrix   $\bm{P}_{E_k}\!\!=\!\!\mathbb{E}[\bm{x}_{E_k}\bm{x}_{E_k}^H]\!\!\in\!\mathbb{C}^{N_U \times N_U} $ to the H-AP, respectively. The transmission power at the IR is fixed as $P_I$. Since the energy at the ER $E_k$ is only coming from energy harvesting in the first phase, the transmission energy at the ER $E_k$ should not exceed the harvested energy $Q_{E_k}$, i.e., $(1-\tau_0)\text{tr}(\bm{P}_{E_k})\!\le \! \tau_0\varepsilon_k\text{tr}(\bm{H}_{E_k}\bm{W}_B\bm{H}_{E_k}^H),~\forall k \!\in \!\mathcal{K}$. By denoting $\bm{g}_{I_R}\!\in\!\mathbb{C}^{ N_B}$  and $\bm{G}_{E_k}\!\in\!\mathbb{C}^{N_B \times N_U}$ as the uplink channels from the IR  and ER $E_k, \forall k \!\!\in \!\!\mathcal{K}$ to  the H-AP, respectively, we have the received signal at the H-AP as
 \begin{align}\label{eq4}
   \bm{y}_{B}=
 \sum\limits_{k=1}^K\bm{G}_{E_k}
 \bm{x}_{E_k}+ \bm{g}_{I_R}\sqrt{{P}_{I}}
 {s}_{R}+\bm{n}_{B},
 \end{align}
where $\bm{n}_{B}\sim \mathcal{CN}(\bm{0}, \sigma_n^2\bm{I}_{N_B})$ denotes the received AWGN noise at the H-AP. Similarly to \cite{succ},  we assume that successive interference cancellation technique  is  adopted at the  H-AP, and  thus the achievable uplink sum rate of the SDMA-enabled WPSN can be formulated as
 \begin{align}\label{eq5}
 R_{S}^{U}\!&=\!
 (1\!\!-\!\!\tau_0)\log\det\big(\widetilde{\bm{I}}_R\!+\! \sigma_n^{-2} \sum\limits_{k=1}^K \bm{G}_{E_k}
   \bm{P}_{E_k}\bm{G}_{E_k}^H\big)
 \end{align}  where $\widetilde{\bm{I}}_R=\bm{I}_{N_B}\!+\! \sigma_n^{-2} P_I\bm{g}_{I_R} \bm{g}_{I_R}^H$. In this paper, we aim at maximizing the uplink sum rate $R_{S}^{U}$ in \eqref{eq5} while guaranteeing the downlink communication quality-of-service (QoS), i.e., the downlink information rate $R_{I_R}^D$ in \eqref{eq3}, by jointly optimizing
 the  time splitting ratio $\tau_0$, downlink energy beamforming $\bm{W}_B$  and
 uplink information beamforming $\bm{P}_{E_k}, \forall k \!\!\in \!\!\mathcal{K}$.
Mathematically, the uplink sum rate maximization  problem in the SDMA-enabled WPSN is formulated as \vspace{-3mm}
 \begin{align}\label{eq6}
 &\underset{\substack{{\tau}_0,\bm{W}_B\succeq\bm{0},\\
 \bm{P}_{E_k}\succeq\bm{0}, \forall k} }
{\text{max}}~(1\!\!-\!\!\tau_0)\log\det\big(\widetilde{\bm{I}}_{R}\!+\! \sigma_n^{-2} \!\sum\limits_{k=1}^K \bm{G}_{E_k}
   \bm{P}_{E_k}\bm{G}_{E_k}^H\big)\nonumber\\
&{\rm{s.t.}}~\text{CR1:}~\text{tr}(\bm{W}_B) \le P_B,~0\le \tau_0\le 1, \nonumber\\
&~~~~~\text{CR2:}~ \tau_0\log(1\!+
 \!\sigma_n^{-2}\bm{h}_{I_R}^H\bm{W}_B
 \bm{h}_{I_R})\ge R_{I}, \nonumber\\
 &~~~~~\text{CR3:}~(1\!\!-\!\!\tau_0)\text{tr}(\bm{P}_{E_k})\!
 \!\le \!\!\tau_0 \varepsilon_k\text{tr}(\bm{H}_{E_k}\bm{W}_B\bm{H}_{E_k}^H),\forall k .
\end{align}
Here, the constraint $\text{CR1}$ corresponds to the maximum transmission power constraint at the H-AP, $\text{CR2}$ models the downlink QoS constraint at the IR node where $R_I$ denotes the required minimum downlink rate, and $\text{CR3}$ considers the uplink transmission energy constraints at the ERs since their uplink energies are only coming from the harvested energies $Q_{E_k}$ \eqref{eq02} in the downlink.
Notice that our uplink sum rate maximization is different from those for fully  WPSNs in \cite{RWPCN2,RWPCN3,RWPCN30,RWPCN00,RWPCN0,RWPCN1, RWPCN5}, since both energy and information transfers are conducted simultaneously in the downlink and additional downlink rate constraint $\text{CR2}$ is considered to guarantee the QoS for the information transfer to the IR. Moreover, our rate maximization problem also differs from that in \cite{RWPCN6} with additional uplink energy constraints $\text{CR3}$. Clearly, the downlink rate  constraint and the uplink energy constraints are practical and necessary to be considered in  partially  WPSNs. However, their consideration complicates the optimization problem with highly coupled variables $\{{\tau}_0,\bm{W}_B,\bm{P}_{E_k}, \forall k \!\!\in \!\!\mathcal{K}\}$, and  the problem becomes non-convex and difficult to solve.
%
 \subsection{ TDMA-enabled WPSN }
For the TDMA-enabled WPSN,  the received signals  at the ER $E_k$ and the IR in the downlink phase are the same as that in \eqref{eq1} and \eqref{eq2}, respectively. However, in  the  uplink phase,
 the IR and each ER $E_k$ sequentially transmit signals to the H-AP within the slots of $\tau_{I_R}$ and $\tau_{E_k}$, respectively. Therefore,
the achievable uplink sum rate  is  given by\vspace{-1mm}
\begin{align}\label{eq10}
 R_{T}^U&= \sum\limits_{k=1}^K \tau_{E_k}\log\det(\bm{I}_{N_B}+ { \sigma_n^{-2}} \bm{G}_{E_k}
   {\bm{P}}_{E_k}\bm{G}_{E_k}^H)\nonumber\\
  & +\tau_{I_R}\log(1+ { \sigma_n^{-2}} P_I \bm{g}_{I_R}
   \bm{g}_{I_R}^H
   )
\end{align}
Accordingly, the uplink sum rate maximization
problem in the TDMA-enabled WPSN is formulated as
\begin{align}\label{eq12}
 &\underset{\substack{\bm{\tau},\bm{W}_B\succeq\bm{0},
 \bm{P}_{E_k}\succeq\bm{0}}}
{\text{max}}~\sum\limits_{k=1}^K \tau_{E_k}\log\det(\bm{I}_{N_B}+ { \sigma_n^{-2}} \bm{G}_{E_k}
   {\bm{P}}_{E_k}\bm{G}_{E_k}^H)\nonumber\\
  &~~~~~~~~~~~~ +\tau_{I_R}\log\det(1+ { \sigma_n^{-2}} P_I \bm{g}_{I_R}
   \bm{g}_{I_R}^H
   )\nonumber\\
&{\rm{s.t.}} ~\text{CR1:}~\text{tr}(\bm{W}_B) \le P_B,~0\le \tau_0\le 1, \nonumber\\
&~~~~~\text{CR2:}~ \tau_0\log(1\!+
 \!\sigma_n^{-2}\bm{h}_{I_R}^H\bm{W}_B
 \bm{h}_{I_R})\ge R_{I}, \nonumber\\
&~~~~~\text{CR4:}~\tau_{E_k}
\text{tr}(\bm{P}_{E_k})\!\!
 \le \!\!\tau_0 \varepsilon_k\text{tr}(\bm{H}_{E_k}\!\bm{W}_B\!
 \bm{H}_{E_k}^H),\forall k,\nonumber\\
 &~~~~~\text{CR5:}~\tau_{I_R} + \sum\limits_{k=1}^{K}\tau_{E_k}=1-\tau_0
\end{align}
where $\bm{\tau}\!\!=\!\![\tau_0,\tau_{E_1},\cdots,\tau_{E_K},\tau_{I_R}]$. Similarly to $\text{CR3}$ in \eqref{eq6}, the constraint $\text{CR4}$ models the uplink transmission energy constraints but with the uplink slot for the ER $E_k$ adjusted as $\tau_{E_k}$. Additionally, the constraint $\text{CR5}$ is introduced due to the implementation of TDMA protocol.
Clearly, the  problem \eqref{eq12} is also nonconvex  and is more challenging than the
problem \eqref{eq6}
since more time slots are involved in the vector $\bm{\tau}$ for optimization. Obviously, the key challenge in the problems \eqref{eq6} and \eqref{eq12} comes from the downlink rate constraint $\text{CR2}$, its feasibility will be discussed first before we proceed to solve for the optimization problems.

\section{ Feasibility Analysis }\label{sectionIII}
\subsection{Feasibility of  downlink rate constraint}
Different from most of the prior designs for WPSNs, downlink rate constraint $\text{CR2}$ is considered here due to the adoption of simultaneous wireless power and information transfer in the downlink. Clearly from \eqref{eq6}
and \eqref{eq12}, the feasible downlink rate threshold $R_I$ is limited by the achievable  downlink rate $R_{I_R}^D $ in \eqref{eq3} which depends on the downlink beamforming $\bm{W}_B$, while $\bm{W}_B$ is also constrained by the maximum downlink transmission power in $\text{CR1}$. The feasible downlink rate threshold is thus upper bounded by
%
\begin{align}\label{eq20}
&\underset{\tau_0,\bm{W}_B\succeq\bm{0}}
{\text{max}}~R_I\nonumber\\
 &{\rm{s.t.}}~ \text{CR1:}~~\text{tr}(\bm{W}_B) \le P_B,~0\le \tau_0\le 1, \nonumber\\
&~~~~~\text{CR2:}~ \tau_0\log(1\!+
 \!\sigma_n^{-2}\bm{h}_{I_R}^H\bm{W}_B
 \bm{h}_{I_R})\ge R_{I}.
\end{align}
It is easily observed that the optimal $\tau_0$ for the problem \eqref{eq20} is $\tau_0\!=\!1$, and then the problem \eqref{eq20} reduces to a conventional
rate maximization problem for a MISO system and has been
solved in \cite{RR2}. More specifically,  the
optimal solution of $\bm{W}_B$ for the  problem \eqref{eq20} is $\bm{W}_B^{up}=\frac{P_B}
{\Vert\bm{h}_{I_R}\Vert^2}\bm{h}_{I_R}\bm{h}_{I_R}^H$, which means that the downlink beamforming is aligned for information transmission only. The upper bound of feasible downlink rate threshold is correspondingly given as
 $R_I^{up}=\log(1\!+\!\sigma_n^{-2}P_B\Vert
\bm{h}_{I_R}\Vert^2)$ \cite{RR2}. Since both  SDMA-enabled and TDMA-enabled WPSNs have the same downlink process, this upper bound $R_I^{up}$ is applicable for both problems \eqref{eq6} and \eqref{eq12}. In other words, we can conclude that when $R_I \in [ 0, R_I^{up}]$,  the  problems \eqref{eq6} and  \eqref{eq12} are feasible.

\subsection{Tightness of downlink rate constraint}
In the feasible region $R_I\in [0, R_I^{up}]$, the tightness of downlink rate constraint $\text{CR2}$ depends on the actual rate threshold $R_I$ and will heavily affect the optimal solutions for the problems \eqref{eq6} and  \eqref{eq12}. Now we take the problem \eqref{eq6} for tightness investigation first. Specifically, if neglecting  the
downlink rate  constraint ${\text{CR}}2$, the problem
\eqref{eq6} reduces to the
sum  rate maximization problem for fully WPSNs in \cite{RWPCN30},
 and   the corresponding
optimal downlink beamforming is  given as ${\bm{W}}_B^{mi}\!\!=\!\!P_B\bm{u}_B^{mi}(\bm{u}_B^{mi})^H$, where $\bm{u}_B^{mi}\!\!\in\! \!\mathbb{C}^{N_B}$  is  the dominated eigenvector of a certain
 linear combination of  the covariance matrices for all ERs' downlink channels, i.e., $\bm{H}_{E_k}\bm{H}_{E_k}^H, \forall k \!\!\in \!\!\mathcal{K}$  \cite{RWPCN30}. This solution means that the downlink beamforming is aligned for energy transfer only.

 Given this downlink beamforming, the achievable downlink rate can be expressed as $R_I^{mi}\!\!=\!\!\log(1\!\!+\!\!\sigma_n^{-2}
\bm{h}_{I_R}^H{\bm{W}}_B^{mi}\bm{h}_{I_R})$, and the corresponding optimal achievable uplink sum rate is in fact an upper bound of that of the problem \eqref{eq6} with the constraint ${\text{CR}}2$. If the downlink rate threshold $R_I$ is less than the rate $R_I^{mi}$, i.e., $R_I < R_I^{mi}$, the downlink rate constraint ${\text{CR}}2$ can be automatically satisfied with the inequality strictly holding, i.e., $\tau_0\log(1\!+
 \!\sigma_n^{-2}\bm{h}_{I_R}^H\bm{W}_B
 \bm{h}_{I_R}) > R_{I}$. In other words, when $R_I \in [0, R_I^{mi})$, the constraint ${\text{CR}}2$ is inactive and can be ignored in the problem \eqref{eq6}. However, when $R_I^{mi}\le R_I\le R_I^{up}$, the downlink rate constraint cannot be ignored and we have the following result.
 \begin{lemma}\label{lemm0}
 When $R_I^{mi}\le R_I\le R_I^{up}$, the
 downlink rate constraint in the problem \eqref{eq6} is tight, i.e., the optimal solution will exist at the boundary with $\tau_0\log(1\!+
 \!\sigma_n^{-2}\bm{h}_{I_R}^H\bm{W}_B
 \bm{h}_{I_R}) = R_{I}$ satisfied.
  \end{lemma}
   \begin{proof}
Please see Appendix A.
\end{proof}
With respect to the problem \eqref{eq12}, by neglecting the constraint ${\text{CR}}2$, although the corresponding optimal downlink beamforming and the achievable downlink rate (${\bm{W}}_B^{mi}$ and $R_I^{mi}$) cannot be directly obtained based on the results in \cite{RWPCN30} due to the extra time allocations $\tau_{I_R}$ and $\tau_{E_k}$, they can be derived using the joint concavity of the problem \eqref{eq12} proved in  Section III. C. Then the above tightness result also holds for the problem \eqref{eq12} in the TDMA-enabled WPSNs. Since the tightness proof is similar to that in Appendix A, it is omitted here for conciseness.

\subsection{Problem reformulation }
In order to  solve  the uplink sum rate
maximization  problem  \eqref{eq6} for the SDMA-enabled WPSNs effectively, we define two new
variables  $\widetilde{\bm{W}}_B=
\tau_0{\bm{W}}_B$ and $\widetilde{\bm{P}}_{E_k}
=(1\!-\!\tau_0) \bm{P}_{E_k}, \forall k \!\in \!\mathcal{K}$. Then the problem \eqref{eq6}  can be reformulated as
\begin{align}\label{eq14}
  &\underset{{{\tau}_0, \{\widetilde{\bm{W}}_B,
\widetilde{\bm{P}}_{E_k}\}\succeq\bm{0}}}
{\text{max}}\!\!(1\!\!-\!\!\tau_0)\log\det\big(\widetilde{\bm{I}}
_{R}\!\!+\!\! \frac{\sigma_n^{\!-\!2}}{1\!\!-\!\!\tau_0} \!\sum\limits_{k=1}^K \bm{G}_{E_k}
   \widetilde{\bm{P}}_{E_k}\bm{G}_{E_k}^H\big)\nonumber\\
&{\rm{s.t.}}~\widetilde{\text{CR}}1:~~\text{tr}(\widetilde{\bm{W}}_B) \!\le \! P_B\tau_0,~0\!\le \! \tau_0\!\le\! 1,\nonumber\\
&~~~~~\widetilde{\text{CR}}2:~ \tau_0\log(1\!+\frac{\sigma_n^{-2}}{\tau_0}\!\bm{h
}_{I_R}^H
\widetilde{\bm{W}}_B
 \bm{h}_{I_R})\ge{ R_{I}}\nonumber\\
 &~~~~~\widetilde{\text{CR}}3:~\text{tr}(\widetilde{\bm{P}}_{E_k})\le \varepsilon_k\text{tr}(\bm{H}_{E_k}\widetilde{\bm{W}}_B\bm{H}_{E_k}^H),~\forall k.
\end{align}
Clearly, the objective function of the problem \eqref{eq14} is the perspective  of the
 concave function  $f( \widetilde{\bm{P}}_{E_k}) \!= \!
 \log\det\big(\widetilde{\bm{I}}_{R}
 \!+\!\sigma_n^{-2}\sum\limits_{k=1}^K \bm{G}_{E_k}
   \widetilde{\bm{P}}_{E_k}\bm{G}_{E_k}^H)\big)$. According  to  \cite[p.~39]{Sp},  the concavity is
preserved by the perspective operation. Therefore, the objective function is strictly and jointly
concave with respect to (w.r.t.) $\{{\tau}_0,
\widetilde{\bm{P}}_{E_k},\forall k \!\in \!\mathcal{K}\}$. In addition, all
constraints in \eqref{eq14} are  convex. We then can conclude that the problem  \eqref{eq14} is  jointly  concave w.r.t.
$\{{\tau}_0, \widetilde{\bm{W}}_B,
\widetilde{\bm{P}}_{E_k}, \forall k \!\!\in \!\!\mathcal{K} \}$.
Similarly,  by redefining $ \widetilde{\bm{P}}_{E_k}\!\!=\!\!\tau_{E_k}
\bm{P}_{E_k},\forall k \!\!\in\! \!\mathcal{K}$, the uplink sum
rate maximization  problem \eqref{eq12} for  TDMA-enabled WPSNs can also be
reformulated as
\begin{align}\label{eq15}
 &\underset{\substack{\bm{\tau},\widetilde{\bm{W}}_B\succeq\bm{0},
 \widetilde{\bm{P}}_{E_k}\succeq\bm{0}}}
{\text{max}}~\sum\limits_{k=1}^K \tau_{E_k}\log\det(\bm{I}_{N_B}+ \frac{ \sigma_n^{-2}}{\tau_{E_k}
  } \bm{G}_{E_k}
  \widetilde{\bm{P}}_{E_k}\bm{G}_{E_k}^H)\nonumber\\
&~~~~~~~~~~~~ +\tau_{I_R}\log\det(\bm{I}_{N_B}+  \sigma_n^{-2}P_I\bm{g}_{I_R}\bm{g}_{I_R}^H
   )
   \nonumber\\
   &~~{\rm{s.t.}}~~\widetilde{\text{CR}}1,~
   \widetilde{\text{CR}}2,~
   \widetilde{\text{CR}}3, ~{\text{CR}}5.
\end{align}
Following a similar logic of proving the concavity of the problem \eqref{eq14}, the problem  \eqref{eq15} is also jointly concave w.r.t. $\{\bm{\tau},
\widetilde{\bm{W}}_B,
\widetilde{\bm{P}}_{E_k}, \forall k \!\in \!\mathcal{K} \}.$ Both problems \eqref{eq14}
and  \eqref{eq15}  can be numerically
solved by standard convex optimization technique \cite{Sp}, and the globally optimal beamforming ${\bm{W}}_B$ and ${\bm{P}}_{E_k}$ and the time slots can then be obtained with simple variable substitution. However, the numerical solution not only has high computational complexity but also provides little insight. In the following, we will propose insightful semi-closed-form optimal solutions for the jointly concave problems with low complexity.
\section{Semi-closed-form optimal solutions for SDMA-enabled
and TDMA-enabled WPSNs }
It is noticed that when the time splitting ratio ${\tau}_0$ is given, the problems \eqref{eq14} and  \eqref{eq15} are still jointly concave w.r.t. the other design variables $\{\widetilde{\bm{W}}_B,
\widetilde{\bm{P}}_{E_k}, \forall k \!\!\in \!\!\mathcal{K}\}$. It inspires us to solve for the optimal solution for the other design variables by fixing ${\tau}_0$, and then find the optimal ${\tau}_0$ afterwards.
\subsection{SDMA-enabled uplink sum rate maximization} \label{A}
When ${\tau}_0$ is given, the SDMA-enabled uplink sum rate maximization problem \eqref{eq14}  can be
rewritten as
\begin{align}\label{eq21}
 f_S({\tau}_0)\!= \!&\underset
  {\substack{\widetilde{\bm{W}}_B\succeq\bm{0},  \\ \widetilde{\bm{P}}_{E_k}\succeq\bm{0}, \forall k  }}
 {\text{max}}\!\!\!\!(1\!\!-\!\!{\tau}_0)\log\!\det\!\!\bigg(\!
\widetilde{\bm{I}}_{R}\!\!+\!
\frac{\sigma_n^{-2}}{(1\!\!-\!\!{\tau}_0)} \!\sum\limits_{k=1}^K \bm{G}_{E_k}
  \widetilde {\bm{P}}_{E_k}\bm{G}_{E_k}^H\!\bigg)\nonumber\\
&{\rm{s.t.}}~~~\widetilde{\text{CR}}1,~
   \widetilde{\text{CR}}2,~
   \widetilde{\text{CR}}3,
\end{align}
and  its corresponding Lagrangian function
 is given by
\begin{align}\label{eq22}
   &L_S(\mathcal{A}_S)\!=\!(1\!-\!
 {\tau}_0)\log\det(\bm{M}_{\mathcal{K}\setminus k}^{-1}\!+\! \frac{\sigma_n^{-2}}{(1\!\!-\!\!{\tau}_0)} \bm{G}_{E_k}
  \widetilde {\bm{P}}_{E_k}\bm{G}_{E_k}^H) \nonumber\\
   &\!+\!
   \text{tr}\left((\widetilde{\bm{H}}\!+\!\bm{Z}_0)\widetilde{\bm{W}}_B\right) \!\!-\!\!\sum\limits_{k=1}^K
  \text{tr}((\mu_k\bm{I}_{N_U}
  \!-\!\bm{Z}_k)\widetilde{\bm{P}}_{E_k})
   \!+\!\epsilon_S,
\end{align}
where $\mathcal{A}_S\!\!=\!\!\{\widetilde{\bm{W}}_B,
 \widetilde{\bm{P}}_{E_k}, \bm{Z}_0,
\lambda, \beta,\bm{Z}_k,\mu_k,\forall k\}$  and \vspace{-1mm}
\begin{align}\label{eq23}
&{\bm{M}}_{\mathcal{K}\setminus k}\!\!=\!\!{\bm{M}}_{\mathcal{K}\setminus k}^{\frac{1}{2}}
 {\bm{M}}_{\mathcal{K}\setminus k}^{\frac{1}{2}}\!\!=\!\!(\widetilde{\bm{I}}_{R}
 \!+\!\frac{\sigma_n^{-2}}{(1\!\!-\!\!{\tau}_0)} \sum\limits_{i
\neq k}\bm{G}_{E_i}
   \widetilde{\bm{P}}_{E_i}\bm{G}_{E_i}^H)
   ^{-1}, \forall k,\nonumber\\
  &\widetilde{\bm{H}}\!\!=\!\!{\bm{H}}\!\!-\!\!\lambda
   \bm{I}_{N_B},~{\bm{H}}\!=\!\!\!
   \sum\limits_{k=1}^K\mu_k
   \varepsilon_k \bm{H}_{E_k}^H\bm{H}_{E_k}\! \!
   + \!\! \beta \bm{h}_{I_R}
   \bm{h}_{I_R}^H,\nonumber\\
   & \epsilon_S\!=\!\lambda {\tau}_0 P_B+
  \beta\sigma_n^2{\tau}_0(1-2^{R_I/{\tau}_0}).
\end{align}

Here, $ \{\lambda, \beta$,  $\mu_k, \forall k  \}$   are
the non-negative lagrangian multipliers corresponding to
 constraints $\widetilde{\text{CR}}1$,$\widetilde{\text{CR}}2$ and
$\widetilde{\text{CR}}3$ in the problem \eqref{eq21}, respectively. While $\bm{Z}_0\succeq\bm{0}$ an  $\bm{Z}_k\succeq\bm{0},\forall k$ are the lagrangian multipliers corresponding to
$\widetilde{\bm{W}}_B\succeq\bm{0}$ and $\widetilde{\bm{P}}_{E_k}\succeq\bm{0},\forall k$, respectively.   ${\bm{M}}_{\mathcal{K}\setminus k}^{\frac{1}{2}}$ denotes the Hermitian square root of the positive definite matrix ${\bm{M}}_{\mathcal{K}\setminus k}, \forall k $. Since the problem \eqref{eq21} is concave, its Karush-Kuhn-Tucker (KKT) conditions are necessary and sufficient for the optimal solution. Based on the KKT conditions and the definitions $\widetilde{\bm{W}}_B=
\tau_0{\bm{W}}_B$ and $\widetilde{\bm{P}}_{E_k}
=(1\!-\!\tau_0) \bm{P}_{E_k}, \forall k \!\in \!\mathcal{K}$, the optimal beamforming and lagrangian multipliers should follow the structure shown in the following theorem.
%
%

\begin{proposition}\label{pop1}
 For any given time splitting ratio ${\tau}_0$, the optimal lagrangian multipliers $\lambda^{\star},
 \beta^{\star}, \mu_k^{\star},\forall k$, the optimal
  downlink
 beamforming  $\bm{W}_B^{\star}$ and  the
 optimal uplink
 beamforming $ {\bm{P}}_{E_k}^{\star},
 \forall k$ to the problem \eqref{eq21} are expressed as
 \begin{subequations}
\begin{align}
&(\bm{W}_B^{\star},\lambda^{\star}) \! \! =
\!\!\left\{\!\!\!\begin{array}{ll}
  ( P_B\bm{u}_{{H}}\bm{u}_{{H}}^H,\lambda_{{H}}^{\max}) & 0\!\!\le\! \!R_I \!\!<\!\! R_I^{up}\\
 (\frac{P_B}{\Vert\bm{h}_{IR}\Vert^2}\bm{h}_{IR}\bm{h}_{IR}^H, 0) &R_I\!\!=\!\!R_I^{up}
\end{array}\right.\!\!\!,\label{eq250}\\
& {\bm{P}}_{E_k}^{\star}\!\!=
\!\!\bm{V}_{M_{\mathcal{K}\setminus k}}^{\star  } \! \bm{\Lambda}_{P_{E_k}}\!\bm{V}_{M_{\mathcal{K}\setminus k}}^{\star H },\label{eq251}\\
&\bm
{\Lambda}_{P_{E_k}}\!\!\!=\!\!\text{diag}[ {\Lambda}_{P_{E_k},1},\cdots,  {\Lambda}_{P_{E_k},N_U}]\\
&\Lambda_{P_{E_k},i}\!\!=\!\!\left[\frac{1}{\ln 2\mu_k^{\star}}\!-\!
\frac{\sigma_n^2}
{\Lambda_{M_{\mathcal{K}\setminus k}^{\star},i}^2}\right]^+\!\!
\!,~\forall ~i, ~\forall~ k,\label{eq252}\\
&\mu_k^{\star}\!\!=\!\! \frac{N_U(1\!\!-
\!\!{\tau}_0)}{
\ln 2\left({\tau}_0 \varepsilon_k
\text{tr}(\bm{H}_{E_k}\bm{W}_B^{\star}
\bm{H}_{E_k}^H)\!\!+\!\!\!\sum\limits_{i=1}^{N_U}
\frac{(1\!-\!{\tau}_0)\sigma_n^2}
{\Lambda_{M_{\mathcal{K}\setminus k}^{\star},i}^2}\right)}, \label{eq253}\\
&\beta^{\star}\!\!= \!\!\left\{\!\!\!\begin{array}{ll}
  0& 0\!\le\! R_I \!\le\! R_I^{mi} \\
{\arg} ~\{{f}_{R}\left(\bm{W}_B^{\star}\right)\!\!=\!\!{ R_{I}}\} &R_I^{mi}\!\!<\! R_I\! \!< \!\!R_I^{up}\\
+\infty & R_I\!\!=\!\!R_I^{up}
\end{array}\right.\!\!\!\!, \label{eq254}
\end{align}
\end{subequations}
where $\lambda_{{H}}^{\max}$
and $\bm{u}_{{H}}$ are the maximum
eigenvalue and the corresponding dominated
eigenvector of  ${\bm{H}}^{\star}\!\!
=\!\!\!\sum
\limits_{k\!=\!1}^K\mu_k^{\star}
\varepsilon_k
\bm{H}_{E_k}^H\bm{H}_{E_k}$ $\!+ \beta^{\star}
\bm{h}_{I_R}\bm{h}_{I_R}^H$, respectively,
$\bm{V}_{M_{\mathcal{K}\setminus k}}^{\star} $ is defined as the $N_U$-dimensional
right singular  matrix
of ${\bm{M}}_{\mathcal{K}\setminus k}^{\star  \frac{1}{2}}\bm{G}_{E_k} $ based on the singular value decomposition (SVD) $
   {\bm{M}}_{\mathcal{K}\setminus k}^{{\star} \frac{1}{2}}\bm{G}_{E_k}\!\! =\!\!
   \bm{U}_{M_{\mathcal{K}\setminus k}}^{\star}\bm{\Lambda}_{M_{\mathcal{K}\setminus k}}^{\star}
   \bm{V}_{M_{\mathcal{K}\setminus k}}^{\star H}, \forall k$,  and   the
diagonal matrix $ \bm{
\Lambda}_{{M_{\mathcal{K}\setminus k}}}^{\star}
=\text{diag}[\Lambda_{M_{\mathcal{K}\setminus k}^{\star},1},
\cdots,\Lambda_{M_{\mathcal{K}\setminus k}^{\star},N_U}]$ consists of  $N_U$ singular
values of ${\bm{M}}_{\mathcal{K}\setminus k}^{ \star \frac{1}{2}}\bm{G}_{E_k}$. In addition, ${f}_{R}\left(\bm{W}_B^{\star}
\right)$ denotes the achievable downlink rate
function as ${f}_{R}\left(\bm{W}_B^{\star}
\right)\!=\!{\tau}_0\log(1\!+\!\sigma_n^{-2}\bm{h}_{I_R}^H
\bm{W}_B^{\star}
 \bm{h}_{I_R})$.
\end{proposition}
\begin{proof}
  Please see  Appendix B.
\end{proof}

Theorem~\ref{pop1} provides the semi-closed-form optimal solutions for the problem \eqref{eq21}. By iteratively solving from the lagrangian multipliers as well  as the  downlink and uplink beamforming based on \eqref{eq250}-\eqref{eq254}, the optimal solution can be obtained. The convergence of the iterative calculation and the global optimality of the obtained solution are also guaranteed since the problem \eqref{eq21} is jointly concave. Moreover, from Theorem~\ref{pop1}, we have the following insightful observations.

1) The optimal downlink beamforming $\bm{W}_B^{\star}$ is a rank-1 matrix, whose eigenspace is uniquely determined by the dominant eigenvector in the joint eigenspace spanned by $K$ ERs' downlink channels $\bm{H}_{E_k}, \forall k \!\in \!\mathcal{K}$ and the IR's downlink channel $\bm{h}_{I_R}$. Since SWIPT is conducted in the downlink, the downlink beamforming should provide a good balance between energy transfer and information transmission. Whether the optimal downlink beamforming aligns toward the space of the ERs' channels for energy transfer or that of the IR's channel for information transmission is controlled by the lagrangian multipliers $\{\mu_k^{\star}, \forall k \!\in \!\mathcal{K}, \beta^{\star}\}$ and depends on the downlink rate constraint. Specifically, when the downlink rate constraint is not high, i.e., $0\!\le\! R_I \! \le \! R_I^{mi}$, $\beta^{\star}=0$ holds and the optimal downlink beamforming is fully aligned with the eigenspace of the ERs' channels. In other words, the optimal downlink beamforming is designed only aiming at energy transfer while ignoring the need of information transmission since the required information transmission  can be automatically satisfied.

However, when the downlink rate constraint is high, i.e., $R_I^{mi}\!\!<\! R_I\! \!<\!R_I^{up}$, we have $\beta^{\star}={\arg} ~\{{f}_{R}\left(\bm{W}_B^{\star}\right)\!\!=\!\!{ R_{I}}\}$, the need for information transmission cannot be ignored and the optimal downlink beamforming shifts from the space of the ERs' channels toward that of the IR's channel. When the downlink rate constraint is as high as the maximum rate $ R_I\!\!=\!\!R_I^{up}$, $\beta^{\star}=+\infty$ and the downlink beamforming should be fully aligned with the IR's channel to meet the strict information transmission requirement. To some extent, the optimal lagrangian multiplier $\beta^{\star}$ therefore can be regarded as an indicator of the relativity between the downlink beamforming and the IR's downlink channel. High $\beta^{\star}$  means high relativity between the optimal downlink beamforming and the IR's downlink channel.

2) It is seen from \eqref{eq251}$\sim$\eqref{eq253} that the optimal uplink beamforming ${\bm{P}}_{E_k}^{\star}$ depends on the downlink beamforming $\bm{W}_B^{\star}$ through the lagrangian multiplier $\mu_k^{\star}$. This is due to the fact that the uplink transmission energy in each ER is constrained by its energy harvested in the downlink, i.e., the constraint $\widetilde{\text{CR}}3$. Moreover, the optimal uplink beamforming  ${\bm{P}}_{E_k}^{\star}$ for the $k$th ER $E_k$ is related to other ERs' uplink beamforming, i.e., $\{{\bm{P}}_{E_i}^{\star},\forall i\neq k\}$, through the matrix ${\bm{M}}_{\mathcal{K}\setminus k}^{\star}$. To solve for the coupled uplink beamforming $\{{\bm{P}}_{E_k}^{\star},\forall k\}$ in \eqref{eq251}, the  iterative water-filling procedure \cite{water} can be applied. With the concavity of the problem \eqref{eq21}, the iterative water-filling procedure is guaranteed to converge to the globally optimal $\{{\bm{P}}_{E_k}^{\star}, \forall k\}$. Interested readers can refer to \cite{water} for the detailed iterative process.

3) When $R_I^{mi}\!\!<\! R_I\! \!<\!R_I^{up}$, the optimal lagrangian multiplier $\beta^{\star}$ is determined by the nonlinear function
${f}_R\left(\bm{W}_B^{\star}\right)$. After analysis, we find the function ${f}_R\left(\bm{W}_B\right)$ has the monotonically increasing property as follows.
\begin{lemma}\label{lemm2}
Given the downlink beamforming structure $\bm{W}_B = P_B\bm{u}_{{H}}\bm{u}_{{H}}^H$ in \eqref{eq250}, the function ${f}_R(\bm{W}_B)$ is monotonically increasing   w.r.t. $\beta\!\in\! (0,+\infty)$ and converges to $R_I^{up}$ when $\beta \to +\infty$.
\end{lemma}
\begin{proof}
  Please see Appendix C.
\end{proof}
Based on Lemma~\ref{lemm2}, the optimal $\beta^{\star}$ satisfying ${f}_R(\bm{W}_B^{\star})\!=\! { R_{I}} $ can be uniquely determined by the bisection search. Now with the semi-closed-form optimal solution in Theorem~\ref{pop1} for the problem \eqref{eq21}, our remaining task is to find the optimal time splitting ratio ${\tau}_0$ to achieve the maximum uplink sum rate. Mathematically, it is to solve the problem
 ${\tau}_0^{\star}\!\!=\!\!arg \underset{0\le {\tau}_0\le 1}{ \max}~f_S({\tau}_0) $. Since the objective function $f_S({\tau}_0)$  is  concave w.r.t. ${\tau}_0$, the Golden section search can be  utilized  to find the globally optimal ${\tau}_0$ \cite{golden}.
%

\subsection{TDMA-enabled uplink sum rate maximization}
Due to the involvement of additional uplink time allocation vector, the TDMA-enabled uplink sum rate maximization in \eqref{eq15} is more challenging than the problem \eqref{eq14}. As far as we know,  the joint design of beamforming and time allocation vector for TDMA-enabled WPSNs is rarely discussed in the literature. Here we will follow a similar approach as that for SDMA-enabled WPSNs to solve this challenging problem. To be specific, by fixing the time splitting ratio ${\tau}_0$, the problem \eqref{eq15} can be rewritten as
 \begin{align}\label{eq27}
 f_T({\tau}_0)\!=\!&\underset{
 \substack{\bm{\tau}_{up},\widetilde{\bm{W}}_B\succeq\bm{0},\\
 \widetilde{\bm{P}}_{E_k}\succeq\bm{0},\forall k }}
{\text{max}}~\sum\limits_{k=1}^K \tau_{E_k}\log\det(\bm{I}_{N_B}\!\!+\!\!  \frac{\sigma_n^{\!-\!2}}{\tau_{E_k}} \bm{G}_{E_k}
   \widetilde{\bm{P}}_{E_k}\bm{G}_{E_k}^H)\nonumber\\
  &~~~~~~~~~~~~~+\tau_{I_R}\log\det(\bm{I}_{N_B}+ { \sigma_n^{-2}} P_I \bm{g}_{I_R}\bm{g}_{I_R}^H
   )\nonumber\\
   &{\rm{s.t.}}~~\widetilde{\text{CR}}1, ~
   \widetilde{\text{CR}}2, ~\widetilde{\text{CR}}3,~
   {\text{CR}}5.
 \end{align}
 where $\bm{\tau}_{up}\!\!=\!\![\tau_{E_1},\cdots,\tau_{E_K},
   \tau_{I_R}]$ denotes the uplink time allocation.
Clearly,
 the problem  \eqref{eq27}  is
  jointly concave  w.r.t. $\{\bm{\tau}_{up}, \widetilde{\bm{W}}_B, \widetilde{\bm{P}}_{E_k}, \forall k\}$, and
 the corresponding  Lagrangian function  is expressed as\vspace{-2mm}
\begin{align}\label{eq28}
{L_T}({\mathcal{A}_T})&\!=\!\sum\limits_{k=1}^K \tau_{E_k}(\log\det(\bm{I}_{N_B}\!+\!  \frac{\sigma_n^{-2}}{\tau_{E_k}} \bm{G}_{E_k}
   \widetilde{\bm{P}}_{E_k}\bm{G}_{E_k}^H)\!-\!\gamma)\nonumber\\
  & \!+\!
  \text{tr}\left((\widetilde{\bm{H}}\!+\!\bm{Z}_0)
  \widetilde{\bm{W}}_B\right)\! -\!\sum\limits_{k=1}^K\!
 \text{tr}((\mu_k\bm{I}_{N_U}
  \!\!-\!\!\bm{Z}_k)\widetilde{\bm{P}}_{E_k})\nonumber\\
  &\!+\!\tau_{I_R}(C_R\!-\!\gamma) \!+\!{\epsilon}_T.
\end{align}
where ${\mathcal{A}_T}\!\!=\!\!\{\!\bm{\tau}_{up},
\widetilde{\bm{W}}_B,
 \widetilde{\bm{P}}_{E_k},
\lambda, \beta,\gamma, \bm{Z}_0,\bm{Z}_k, \mu_k,\forall k\!\}$,  $C_R\!\!=\!\!\log\det(\bm{I}_{N_B}$ $+ { \sigma_n^{-2}} P_I \bm{g}_{I_R}\bm{g}_{I_R}^H
   )$ is a  constant denoting the spectral efficiency of the IR's uplink channel, and ${\epsilon}_T={\epsilon}_S-\gamma(
   {\tau}_0-1)$. Here,  $ \{\lambda, \beta$,
   $\mu_k,\forall k$,$\gamma\}$  denote
   the non-negative lagrangian multipliers corresponding to the
   constraints $\widetilde{\text{CR}}1$,
   $\widetilde{\text{CR}}2$, $\widetilde{\text{CR}}3$
   and ${\text{CR}}5$ in the
   problem \eqref{eq27}, respectively.   $\bm{Z}_0\succeq\bm{0}$ an  $\bm{Z}_k\succeq\bm{0},\forall k$ are still the lagrangian multipliers corresponding to
$\widetilde{\bm{W}}_B\succeq\bm{0}$ and $\widetilde{\bm{P}}_{E_k}\succeq\bm{0},\forall k$, respectively.  With the concavity of the problem \eqref{eq27} and based on its KKT conditions, the optimal structures for the variables  $\{\bm{\tau}_{up},
{\bm{W}}_B,
 {\bm{P}}_{E_k},
\lambda, \beta,\gamma, \mu_k,\forall k\}$ can be derived in the following theorem.
 \begin{proposition}\label{pop2}
  For any given time splitting ratio ${\tau}_0$, the optimal
  $\{\bm{W}_B^{\star}$,  $\lambda^{\star}$, $\beta^{\star}\}$
  to  the problem
 \eqref{eq27}  are
 identical to that in {Theorem~\ref{pop1}}, while
  the optimal  $\bm{\tau}_{up}^{\star}$ and  ${\bm{P}}_{E_k}^{\star},
 \forall k$ to  the problem
 \eqref{eq27} as well as the optimal lagrangian multipliers $\mu_k^{\star}, \forall k, \gamma^{\star},$ are given by
\begin{subequations}
\begin{align}
& {\bm{P}}_{E_k}^{\star}\!\!\!=\!\!
\bm{V}_{G_{E_k}} \!\! \bm{\Lambda}_{P_{E_k}}\!\!\bm{V}_{G_{E_k}}^H/\tau_{E_k}^{\star},\label{eq2900}\\
&\bm
{\Lambda}_{P_{E_k}}\!\!\!\!=\!\text{diag}[ {\Lambda}_{P_{E_k}\!,\!1},\!\cdots\!,  {\Lambda}_{P_{E_k}\!,\!N_U}\!]
\label{eq2901}\\
&\Lambda_{P_{E_k},i}\!\!=\!\!\left[\frac{\tau_{E_k}^{\star}}{\ln2
\mu_k^{\star}}\!\!-\!\!
\frac{\sigma_n^2\tau_{E_k}^{\star}}{{\Lambda}_{G_{E_k}, i}^2}\right]^+,~i=1,\cdots, N_U\label{eq2902}\\
&\mu_k^{\star}\!\!=\!\! \frac{N_U\tau_{E_k}^{\star}}{
\ln 2\left({\tau}_0 \varepsilon_k
\text{tr}(\bm{H}_{E_k}\bm{W}_B^{\star}
\bm{H}_{E_k}^H)\!\!+\!\!\sum\limits_{i=1}^{N_U} \frac{\sigma_n^2\tau_{E_k}^{\star}}
{\Lambda_{G_{E_k}\!,i}^2}\right)}, \forall k,\label{eq291}\\
&\bm{\tau}_{up}^{\star}\!\!=\!\!\left\{\!\!\begin{array}{ll}
 {\underset{[(\!{\tau}_{E_k})^{+},\forall k, 0]}{\arg}\!\! \left\{\begin{array}{l}\!\!\!{{g}_T({\tau}_{E_k})={\gamma^{\star}}} \\
 {\!\!\!\sum\limits_{k=1}^{K}\!({\tau}_{E_k})^+\!\!=\!1 \!\!- \!{\tau}_0 } \\ \!\!\!{\tau}_{I_R}=0 \!\!\!\!\end{array}\right\}} & \!\!{{\gamma^{\star}}\!>\!C_R} \\
 {\underset{[({\tau}_{E_k}\!)^{+},\forall k, {\tau}_{I_R}]}{\arg}\!\! \left\{\begin{array}{l}\!\!\!{{g}_T({\tau}_{E_k})=C_R} \\
 {\!\!\!\!\sum\limits_{k=1}^{K}\!({\tau}_{E_k})^+\!\!+
 \!{\tau}_{I_R}\!=\!1 \!\!- \!\!{\tau}_0} \\
 \!\!\!{\tau_{I_R}>0}\!\!\!\end{array}\!\!\!\right\}} & \!\!{{\gamma^{\star}}\!=\!C_R}
\end{array}\right.\label{eq2916}
%
%
\end{align}
\end{subequations}
where  $\bm{V}_{G_{E_k}}
 \!\!\in\!\!\mathbb{C}^{N_U\times N_U},\forall k $ is the right singular matrix  of $\bm{G}_{E_k}$ by performing the SVD $\bm{G}_{E_k}=\bm{U}_{G_{E_k}}
 \bm{\Lambda}_{G_{E_k}}
 \bm{V}_{G_{E_k}}^H$,
the diagonal matrix  $\bm{\Lambda}_{G_{E_k}}=
 \text{diag}[{\Lambda}_{G_{E_k},1},
 \cdots, {\Lambda}_{G_{E_k}, N_U}]$ consists of
  $N_U$  singular values of $\bm{G}_{E_k},\forall k$, and the function
${g}_T(\tau_{E_k})\!\!=\!\!\!\sum\limits_{
i\!=\!1}^{N_U} \!\left(\log\!\bigg(1\!+
 \!\frac{
 \sigma_n^{-2}\overline{{\Lambda}}_{G_{E_k},i}}{
 \tau_{E_k}}\!\bigg)\!\!-\!\!\frac{ \sigma_n^{-2}\overline{{\Lambda}}_{G_{E_k},i}}{
 \ln2(\tau_{E_k}\!+
 \sigma_n^{-2}\overline{{\Lambda}}_{G_{E_k},i})}
 \right)$ is defined with
$\overline{{\Lambda}}_{G_{E_k},i}\!\!=\!\!{{\Lambda}}_{G_{E_k},i}^2\Lambda_{P_{E_k},i}, \forall i\!\!=\!\!1,\cdots, N_U$.
\end{proposition}
 \begin{proof}
   Please  see Appendix D.
 \end{proof}

From Theorem~\ref{pop2}, we also have the following insightful observations.

1) Since both  SDMA-enabled WPSNs and TDMA-enabled WPSNs have the same  downlink  transmission stage,  the optimal downlink beamforming $\bm{W}_B^{\star}$ for TDMA-enabled uplink sum rate maximization has the same form as that for SDMA-enabled uplink sum rate maximization. Thus the insightful results on $\{\bm{W}_B^{\star}, \beta^{\star}\}$ for SDMA-enabled WPSNs is also applicable for TDMA-enabled WPSNs.

2) From \eqref{eq2900}$\sim$\eqref{eq291}, we can see that the optimal uplink beamforming ${\bm{P}}_{E_k}^{\star}$ at the $k$th ER depends on the downlink beamforming $\bm{W}_B^{\star}$, but is independent to the uplink beamforming of other ERs, i.e.,  $\{{\bm{P}}_{E_i}^{\star},
\forall i\neq k\}$. This is due to the fact that dedicated time slot is allocated to each ER for uplink transmission in TDMA-enabled WPSNs and thus the uplink beamforming design for each ER can be decoupled.

{3) Additional uplink time allocation is required in TDMA-enabled WPSNs. As shown in \eqref{eq2916}, to maximize the uplink sum rate, there are two possible uplink time allocations depending on the channel conditions of all sensor nodes. Particularly,  the function ${g}_T(\tau_{E_k})$ can also  be roughly regarded as the uplink spectral efficiency of the $k$th ER. Since the uplink spectral efficiency of IR is a constant as $C_R\!=\!\log\det(\bm{I}_{N_B}$ $+ { \sigma_n^{-2}} P_I \bm{g}_{I_R}\bm{g}_{I_R}^H
)$, the optimization of the uplink time slot $\tau_{I_R}$ for  IR  is actually a linearly constrained linear programming problem.  Specifically, when the spectral efficiencies of all  ERs ${g}_T(\tau_{E_k}^{\star})$ are higher than that of the IR $C_R$, it is naturally to allocate the  total uplink time resource (i.e., $1-\tau_0$) only to the ERs, which are  presented in the case of ${{\gamma^{\star}}\!>\!C_R} $  in   \eqref{eq2916}, while no time slot is allocated to the IR, namely $\tau_{I_R}^{\star}=0$.  Otherwise, in order  to guarantee  the nonzero uplink time allocation $\tau_{I_R}^{\star}$ for  the IR, there is at least  a ER achieving the same spectral efficiency as that of the IR, namely, ${{g}_T({\tau}_{E_k}^{\star})=C_R}$.  Thus the total  uplink time  resource (i.e., $1-\tau_0$) is optimally allocated  as that in the case of ${{\gamma^{\star}}\!=\!C_R} $ of  \eqref{eq2916}. Overall, for both two cases, the uplink  time allocation is somehow similar to the water-filling procedure and the spectral efficiency can be regarded as water level to be optimized higher than or equal to $C_R$. To find the optimal uplink time allocation, the nonlinear equation ${{g}_T({\tau}_{E_k})={\gamma^{\star}}}$ needs to be solved. Fortunately, the function ${g}_T({\tau}_{E_k})$ has the following property which can facilitate the derivation of the optimal time slot ${\tau}_{E_k}$.}
%
{\begin{lemma}\label{lemm4}
The function ${g}_T(\tau_{E_k})
 $  is monotonically decreasing w.r.t.   $\tau_{E_k}> 0, \forall k$.
\end{lemma}}

\begin{proof}
Taking the first order derivation of ${g}_T(\tau_{E_k})$ on $\tau_{E_k}$,
we have
$\nabla_{\tau_{E_k}}{g}_T(\tau_{E_k})\!\! =\!\!
\sum\limits_{
i\!=\!1}^{N_U} \frac{-\!(\sigma_n^{-2}\overline{{\Lambda}}_{G_{E_k},i})^2}{\ln2( \tau_{E_k}\! +\!\sigma_n^{-2}\overline{{\Lambda}}_{G_{E_k},i})^2
\tau_{E_k}}$.
It is readily observed that for $\tau_{E_k}\!\!>\!\!0$,
we have  $\nabla_{\tau_{E_k}}{g}_T(\tau_{E_k})\!\!<\!\!0, \forall k$. Therefore,
 the function ${g}_T(\tau_{E_k}),\forall k$ is monotonically decreasing w.r.t. $\tau_{E_k}\!\!>\!\!0$.
   \end{proof}
  With the monotonicity of ${g}_T({\tau}_{E_k})$, the solution of ${g}_T({\tau}_{E_k})
 \!=\!{\gamma^{\star}}$ can be uniquely determined. Moreover, since  $l_T({\tau}_{E_k},\forall k)\!\!=\!\!\!\sum\limits_{k=1}^{K}
 \!{\tau}_{E_k}$ is also monotonically increasing
  w.r.t. ${\tau}_{E_k}, \forall k$, the optimal solution $\bm{\tau}_{up}^{\star}$ and the optimal $\gamma^{\star}$ in \eqref{eq2916} can be efficiently obtained by the iterative bisection search  (IBS)  \cite{IBS}.
%

Similarly to the problem \eqref{eq21}, Theorem~\ref{pop2} also provides a semi-closed-form optimal solution for the problem \eqref{eq27} with a given $\tau_0$. Since the objective function $f_T({\tau}_0)$ is also concave w.r.t. ${\tau}_0$, the Golden section search can still
be applied to find the optimal time splitting ratio $\tau_0^{\star}$ satisfying  ${\tau}_0^{\star}\!=\!arg \underset{0\le {\tau}_0\le 1}{ \max}~f_T({\tau}_0) $.

%
  \begin{algorithm}[t]
\caption{Optimization of the time splitting ratio $\tau_0$} \label{algorithm_1}
\begin{algorithmic}[1]
\STATE Initialize: $\tau_{min}=0$,
$\tau_{max}=1$ and step $\phi=(\sqrt{5}-1)/2$
\REPEAT
\STATE  Calculate $\tau_1\!=\!\tau_{max}\!-\!(\tau_{max}\!-\!\tau_{min
})\phi $ and $\tau_2\!=\!\tau_{min}\!+\!(\tau_{max}\!-\!\tau_{min
})\phi$.
\STATE  Obtain $f_{S/T}(\tau_1)$
and $f_{S/T}(\tau_2)$ from Algorithm 2.
\STATE  \textbf{if} {$f_{S/T}(\tau_1)\!\!>\!\!f_{S/T}(\tau_2)$, set  $\tau_{max}\!\!=\!\!\tau_2$.}
\STATE \textbf{else} {set  $\tau_{min}\!\!=\!\!\tau_1$.} 
\UNTIL{$\vert \tau_{max}\!\!-\!\!\tau_{min}\vert\!\!\le \!\!\kappa$, where $\kappa\!\!>\!\!0$ is sufficiently small.}
\RETURN  the optimal $\tau_0^{\star}=(\tau_{max}\!+\!\tau_{min})/2$
\end{algorithmic}
\end{algorithm}
\begin{algorithm}[t]
\caption{Semi-closed-form solutions in {Theorem~\ref{pop1}} and {Theorem~\ref{pop2}}} \label{algorithm_2}
\begin{algorithmic}[1]
\STATE Input:  ${\tau}_0=\tau_1/\tau_2$;  initial $\mu_k^{(0)}$, $ {\bm{P}}_{E_k}^{(0)}\!\!=\!\!\bm{0}, \!\forall k$ and $\bm{\tau}_{up}^{(0)}$;  iteration index $i\!=\!0$.
\REPEAT
\STATE  Given $\mu_k^{(i)},\forall k$, apply the bisection search to  \eqref{eq254} to find $\beta^{(i)}$, then  $\bm{W}_B^{(i)}$ is obtained  from  \eqref{eq250}.
\STATE \textbf{if} {(SDMA-enabled WPSN is considered)}
\STATE Given $\bm{W}_B^{(i)}$, apply the iterative water-filling procedure to obtain ${\bm{P}}_{E_k}^{(i)}$ and $\mu_k^{(i+1)},\forall k$ from \eqref{eq251}$\sim$\eqref{eq253}.
\STATE\textbf{elseif} ({TDMA-enabled WPSN is considered})
\STATE Given $\bm{W}_B^{(i)},\bm{\tau}_{up}^{(i)}$, obtain  ${\bm{P}}_{E_k}^{(i)}$ and $ \mu_k^{(i+1)}, \forall k$ from \eqref{eq2901}$\sim$\eqref{eq291}.
\STATE Given  ${\bm{P}}_{E_k}^{(i)},\! \forall k$, obtain  $\bm{\tau}_{up}^{(i)}$ from  \eqref{eq2916}.
\STATE \textbf{end}
\STATE  Calculate $f_{S/T}^{(i)}({\tau}_0)$ and update $i\!\!=\!\!i\!+\!1$;
\UNTIL  $f_{S/T}^{(i)}(\!{\tau}_0\!)$ converges.
\RETURN  Optimal $\!{\bm{W}}_B^{(\!i\!)}
,\bm{\tau}_{up}^{(\!i\!)},{\bm{P}}_{E_k}^{(\!i\!)}
, \forall k, f_{S/T}^{(\!i\!)}(\!{\tau}_0\!)$.
\end{algorithmic}
\end{algorithm}

\subsection{Summary and discussion}
For clarification, the implementation of the proposed semi-closed-form optimal solutions for both SDMA-enabled and TDMA-enabled uplink sum rate maximization  is summarized as Algorithm 1 \& 2. Specifically, the Golden section search for the optimal time splitting ratio ${\tau}_0^{\star}$ is shown in Algorithm 1, while the semi-closed-form solutions for the SDMA-enabled problem \eqref{eq21} and the TDMA-enabled problem \eqref{eq27} are illustrated in Algorithm 2.
%

Since the problems  \eqref{eq6} and \eqref{eq12} are jointly concave on $\{{\tau}_0,{\bm{W}}_B,  {\bm{P}}_{E_k}, \forall k\}$ and $\{\bm{\tau},{\bm{W}}_B,  {\bm{P}}_{E_k}, \forall k\}$, respectively, their KKT conditions are sufficient and necessary for the globally optimal solutions no matter whether $\tau_0$ is given. Therefore the semi-closed-form solutions derived based on KKT conditions in Theorem \ref{pop1} and \ref{pop2} are globally optimal. In other words, Algorithm 2 is guaranteed to converge to the global optimal solutions for both  SDMA-enabled problem \eqref{eq21} and the TDMA-enabled problem \eqref{eq27}. Meanwhile, with the concavity of the objective functions $f_{S/T}({\tau}_0)$  w.r.t. ${\tau}_0$, the global optimality of the obtained ${\tau}_0^{\star}$  from  Algorithm 1 is also assured. As a result, our proposed algorithm is guaranteed to converge to the  globally optimal solutions for both  SDMA-enabled and TDMA-enabled uplink sum rate maximizations.

Although iterative calculations are required in our proposed algorithm, the complexity of our proposed algorithm with semi-closed-form solutions is still  lower than that of the traditional numerical algorithms for   standard convex problems, such as, the interior point method. Here, we take the SDMA-enabled WPSN as an example for complexity analysis. As shown in Algorithm 1\&2, Golden section search,   bisection search, iterative water-filling  procedure and SVD/EVD operations   are clearly  involved. By denoting the  converged numbers of iterations for the first three processes  as $I_G$, $I_{B}$ and  $I_W$, respectively,  the complexity of our proposed algorithm can be expressed as $I_GI_{semi}(I_{W}\mathcal{O}(KN_U^3)\!+\!I_B\!+\!\mathcal{O}(N_B^3))$, where $I_{semi}$ denotes the number of iterations in Algorithm 2 and  $\mathcal{O}(n^3)$ denotes the complexity of  SVD/EVD operations.
It is well-known that   Golden section search and the  bisection search are  efficient with small numbers of iterations. Mathematically, we have $I_G\!\!=\!\!\log_2(\frac{1}{\epsilon})$ and $I_B\!\!=\!\!\log_2(\frac{\beta^{max}}{\epsilon})$, where $\epsilon$ denotes the  search accuracy \cite{golden}. In addition, the iterative water-filling procedure usually converges fast with small $I_{W}$ \cite{water} and the number of iterations $I_{semi}$ of Algorithm 2 is also small as shown in the simulations. Nevertheless, by referring  to \cite{complexity}  and recalling the  original jointly concave problem \eqref{eq14}, which is like a SDP problem due to  the positive semidefinite optimization variables, the corresponding complexity of  the numerical interior point method is  $\mathcal{O}(K(KN_U+N_B)^{3.5}+
K^2(KN_U+N_B)^{2.5}+ K^3 (KN_U+N_B)^{0.5})\log(1/\epsilon)$  \cite{complexity}. Clearly, our proposed algorithm has much lower complexity than the numerical convex  algorithm, not to mention the meaningful insights found from our semi-closed-form solutions.

Last but not least, due to the additional  practical consideration of the downlink rate constraint in partially WPSNs, our joint downlink and uplink  beamforming designs as well as  time splitting  are remarkably different  from that in most of the existing works. Particularly, our SDMA-enabled design in fact reduces to that in \cite{RWPCN30} when the downlink rate constraint is not tight. While for the TDMA-enabled WPSNs, due to  another  additional involvement of uplink time allocation, the joint beamforming design and time splitting is more challenging and  rarely discussed in the literature. However,  our proposed semi-closed-form design provides the optimal solutions with low complexity for this challenging problem.

\section{Simulation Results and Discussions} \label{simulation}

In this section, simulation results are provided to demonstrate the effectiveness of the proposed uplink sum rate maximizations for the SDMA-enabled  and TDMA-enabled  MIMO WPSNs, respectively. Unless otherwise stated,  we consider one H-AP with $N_B=6$ antennas,  $K=3$ ERs each with $N_U=3$ antennas and one IR with single antenna. Besides, the H-AP is assumed to locate at the origin $(0,0)$m, while the three ERs and one IR are randomly located  within a circle with radius $10$m.
All wireless  channels are generated according to Rayleigh distribution $\mathcal{CN}(\bm{0}, 10^{-3}d^{-\alpha}\bm{I})$, where $d$ denotes the actual  distance  between the H-AP and ERs/IR,  and $\alpha=3$ is the pathloss exponent.  In addition, the received Gaussian noise variance is set to be $\sigma_n^2=-100$dBm.  The maximum  downlink and uplink transmit  powers of H-AP and IR are defined as $P_B=20$dBm and $P_I=5$dBm, respectively. Here we adopt two benchmark designs for comparisons. One is the maximum downlink rate (MDR) based beamforming scheme where the downlink beamforming is fixed as $\bm{W}_B^{up}\!\!=\!\!\frac{P_B}
{\Vert\bm{h}_{I_R}\Vert^2}\bm{h}_{I_R}\bm{h}_{I_R}^H$, while the other is the proposed uplink sum rate optimization with fixed time allocation $\tau_0=0.5$, which indicates equal time splitting for downlink  and uplink transmissions.

The optimal time splitting $\tau_0$ and the optimal downlink and uplink beamformings are obtained through Golden section search and iterative optimization procedure, respectively,
as shown in Algorithms 1 \& 2. We first investigate the convergence of the proposed  iterative optimization
procedure in Algorithm 2, where  $\tau_0=0.5$  and
$R_I=R_I^0$ satisfying $ R^{mi}< R_I^0
\le R^{up} $ are defined. In particular, two initial values of  $ \bm{W}_B$ are adopted  as follows:
\begin{align}\label{ini}
  &\bm{W}_B^{1}\!=\!\!\frac{P_B}{N_B}\bm{I}_{N_B};~~\bm{W}_B^2\!= \!\bm{a}{\bm{a}^H}, \nonumber\\
  & \vert\bm{h}_{I_R}^H\bm{a}\vert^2
  \!\!=\!\!\sigma_n^2(2^{\frac{R_I^0}{\tau_0}}\!\!-\!1),~~
  \Vert\bm{a}\Vert^2=P_B.
\end{align}
The results are shown in Fig. \ref{Fig1}. It is clear that using both initial $\bm{W}_B$ in \eqref{ini}, the proposed iterative optimization procedure
converges to  the maximum uplink sum rate within 8 iterations for both  SDMA-enabled and TDMA-enabled WPSNs with fixed time splitting $\tau_0=0.5$. Then Fig.~\ref{Fig11} shows  the convergence of  Golden  section search in Algorithm 1  for finding the optimal $\tau_0$. It is also clearly seen that for both WPSNs, the achievable maximum uplink sum rate converges within  5 iterations.
\begin{figure}[t]
\centering
\includegraphics[width=2.42in]{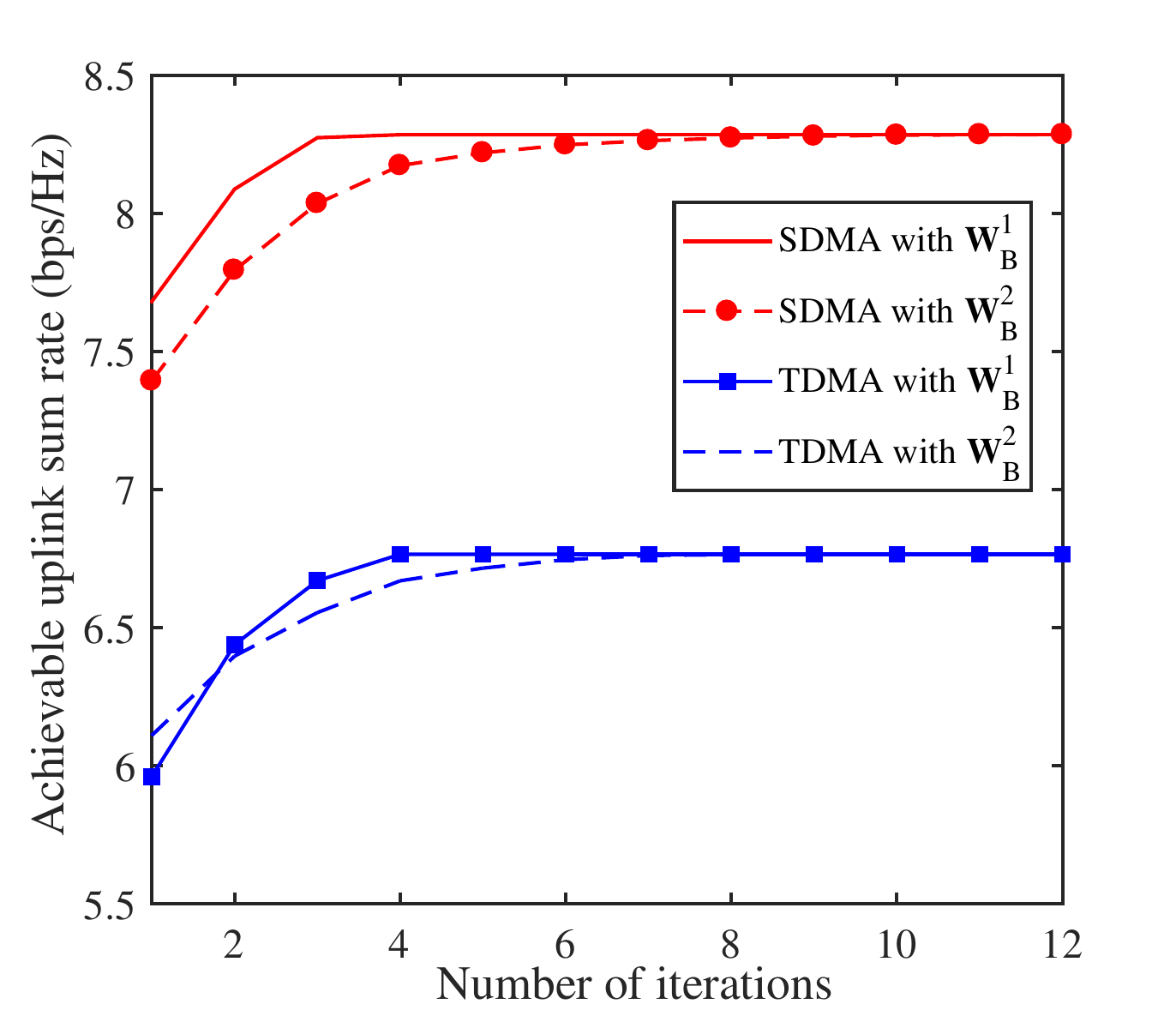}
\caption{Convergence of   Algorithm 2 for both  WPSNs, where
$R_I\!=\!R_I^0$  and  $\tau_0\!=\!0.5$.} %
  \label{Fig1}
\end{figure}
\begin{figure}[t]
\centering
\includegraphics[width=2.42in]{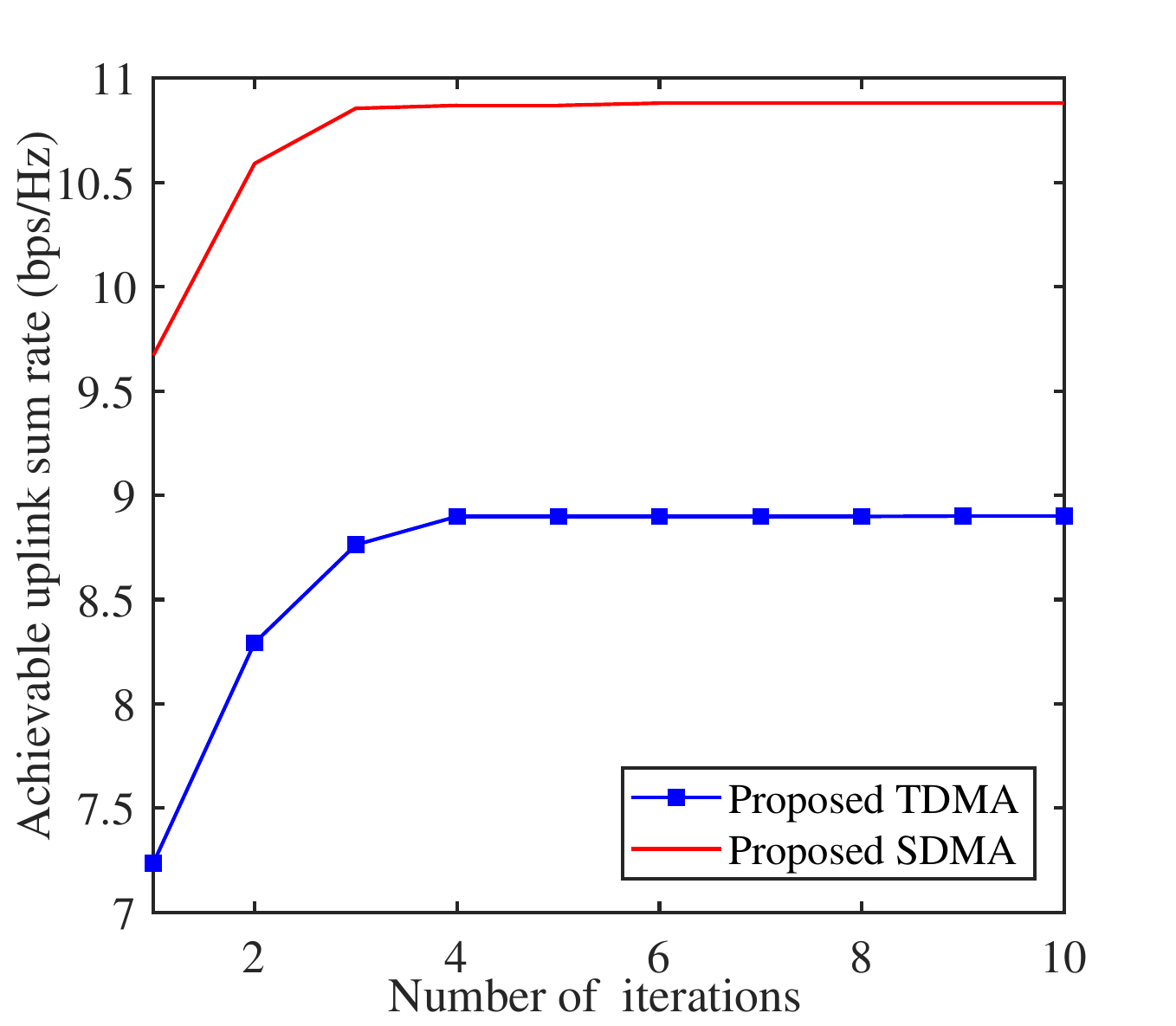}
\caption{ Convergence of   Algorithm 1 for both  WPSNs, where
$R_I\!\!=\!\!R_I^0$.} %
  \label{Fig11}
\end{figure}
\begin{figure}[t]
\centering
\includegraphics[width=2.34in]{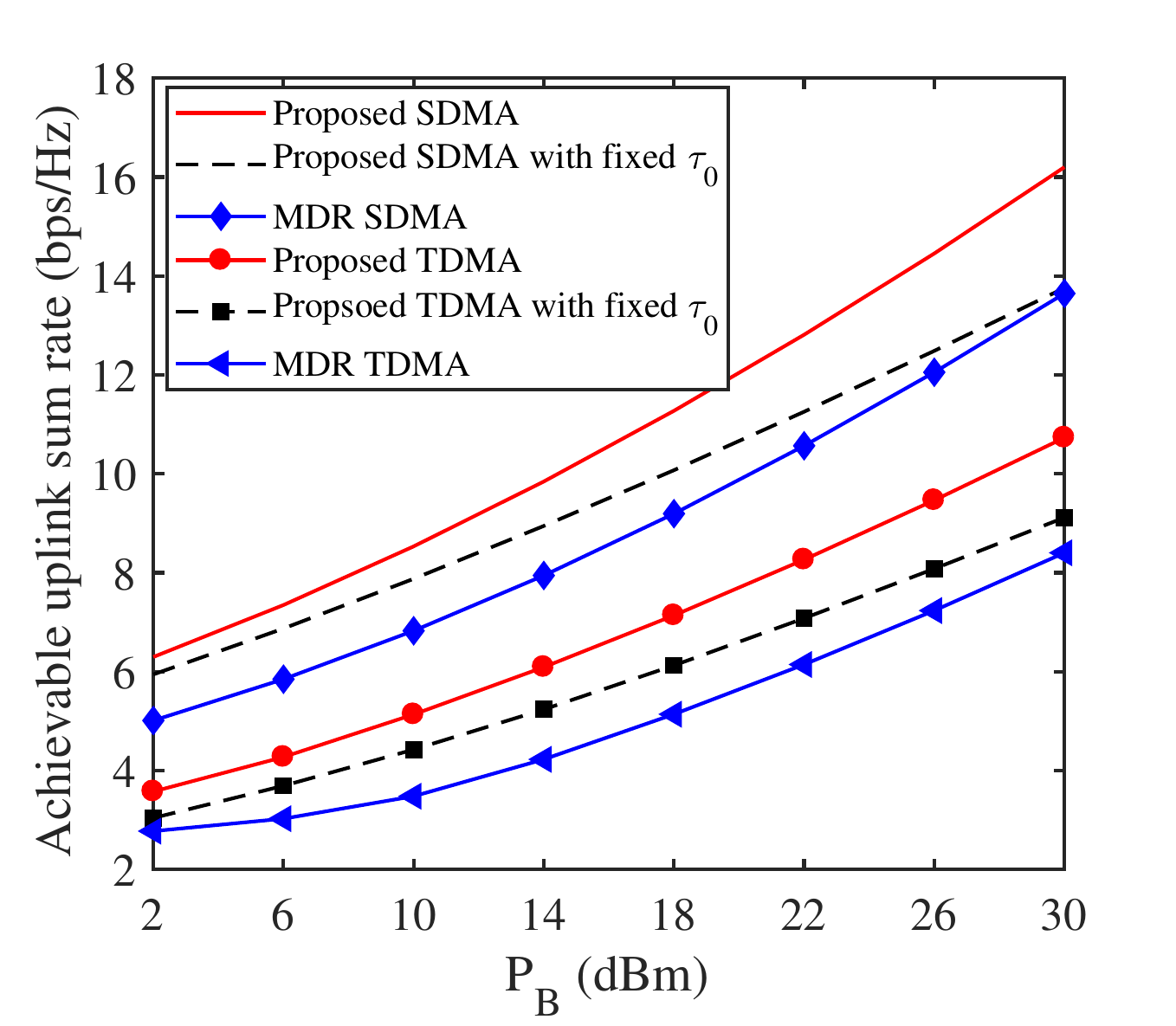}
\caption{Achievable uplink sum rate versus  H-AP transmit power $P_B$ for both  WPSNs, where $R_I\!=\!\frac{R^{mi}}{2}$.} %
  \label{Fig3}
\end{figure}
\begin{figure}[t]
\centering
\includegraphics[width=2.35in]{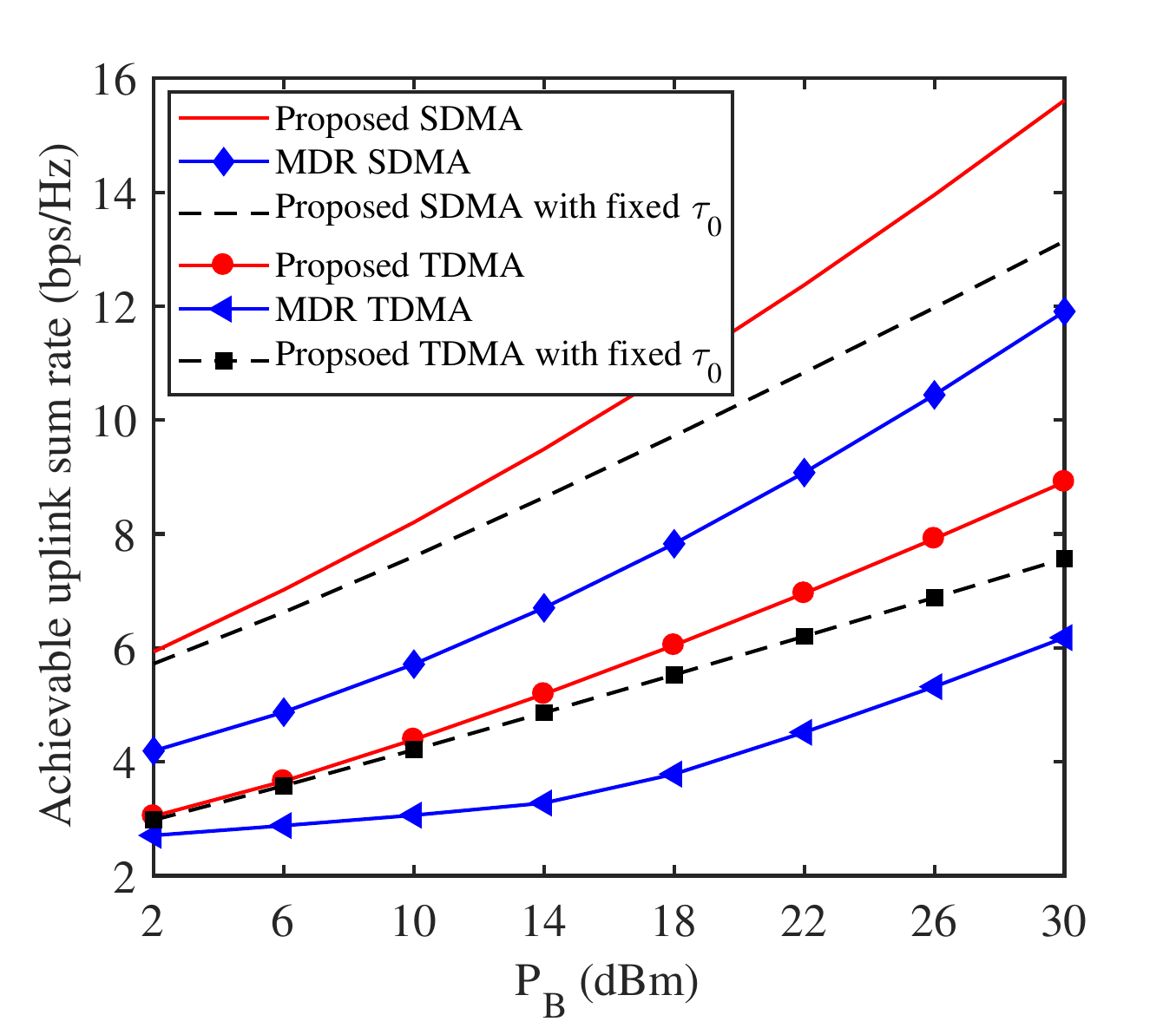}
\caption{Achievable uplink sum rate versus  H-AP transmit power $P_B$ for both  WPSNs, where $R_I\!=R_I^1$ and $ R^{mi}\!<\! R_I^1
\!\le \!0.5R^{up} $.} %
  \label{Fig32}
\end{figure}

In Fig.~\ref{Fig3}, the achievable  uplink sum rate versus  the H-AP transmit power $P_B$ is studied for all three  considered schemes  in both SDMA-enabled and TDMA-enabled WPSNs. Here,  $R_I=\frac{R^{mi}}{2}$  is assumed, which means that  the downlink rate constraint is inactive for  both SDMA-enabled and TDMA-enabled  uplink sum rate maximization. Under this setting, the proposed SDMA-enabled  uplink sum rate maximization  is actually  reduced to the sum throughput maximization for fully  WPSNs in \cite{RWPCN30}. It is observed from Fig.~\ref{Fig3} that the  achievable uplink sum rate  increases with the H-AP transmit power $P_B$ for   all considered schemes.  Besides, it is clear  that the SDMA based scheme generally achieves higher uplink  sum rate than the corresponding  TDMA based scheme due to  the simultaneous uplink transmission  from all the ERs and IR. By comparing the proposed SDMA/TDMA scheme with the corresponding  counterpart with fixed $\tau_0=0.5$,
we  also  find that  the   time splitting optimization plays an important role in improving uplink sum rate  for both SDMA-enabled and TDMA-enabled WPSNs. Additionally, the MDR SDMA/TDMA  scheme performs worst among three SDMA/TDMA schemes due to the fact that most H-AP transmit power is utilized for downlink information transmission, which thus results in  limited harvested energy  at ERs  for uplink transmission.

We then extend this simulation to the case with the downlink threshold  $R_I\!=\!R_I^1$ where $ R^{mi}\!<\! R_I^1
\!\le \!0.5R^{up} $ is defined to guarantee the feasibility of the proposed SDMA/TDMA scheme with fixed $\tau_0=0.5$ and the results are shown in  Fig.~\ref{Fig32}. Notice that under $ R^{mi}\!<\! R_I^1
\!\le \!0.5R^{up} $, the downlink rate constraint is tight and will affect the achievable sum rate and the optimal solutions.
Similar results to that in Fig.~\ref{Fig3} can also  be observed from Fig.~\ref{Fig32}. Moreover, we also find that for all considered schemes,  the  achievable uplink sum rate in Fig.~\ref{Fig32} is naturally lower than that in Fig.~\ref{Fig3} under a given $P_B$, since the uplink sum rate maximization for both WPSNs is constrained by a higher downlink rate threshold $R_I=R_I^1$ in Fig.~\ref{Fig32}.

In Fig.~\ref{Fig4}, we illustrate  the achievable uplink sum  rate   as
a function of the downlink rate threshold $R_I$ for both SDMA-enabled and TDMA-enabled WPSNs.  It is firstly  seen that when $0\le R_I\le R_I^{mi}$, the achievable uplink sum rate is flat, which  indicates  that   the downlink rate constraint  actually  has no   influence  on the achievable   uplink sum rate. Then  when $R_I^{mi}< R_I\le R_I^{up}$, the achievable uplink sum rate decreases with the increase of $R_I$ implying the downlink rate constraint becomes tight. Moreover, we also find that for all three considered schemes, the SDMA-enabled WPSN  has a larger  uplink-downlink rate region   compared to  the TDMA-enabled WPSN. Taking the SDMA-enabled WPSN as an example, it is readily found  that the proposed SDMA scheme with the optimized  $\tau_0$ still performs best,  whereas the MDR SDMA scheme firstly  realizes the lowest uplink sum rate and then  approaches to the proposed SDMA scheme  with the rise of $R_I$, since when $R_I \rightarrow R_I^{up}$, the optimal downlink beamforming $\bm{W}_B^{\star}$ of the proposed SDMA scheme  also tends to be $\bm{W}_B^{up}\!\!=\!\!\frac{P_B}
{\Vert\bm{h}_{I_R}\Vert^2}\bm{h}_{I_R}\bm{h}_{I_R}^H$. As  for  the proposed  SDMA scheme with fixed  $\tau_0=0.5$, the  achievable maximum downlink rate  is easily found to be $0.5R_I^{up}$. The above results also hold for the TDMA-enabled WPSN.

\begin{figure}[t]
\centering
\includegraphics[width=2.4in]{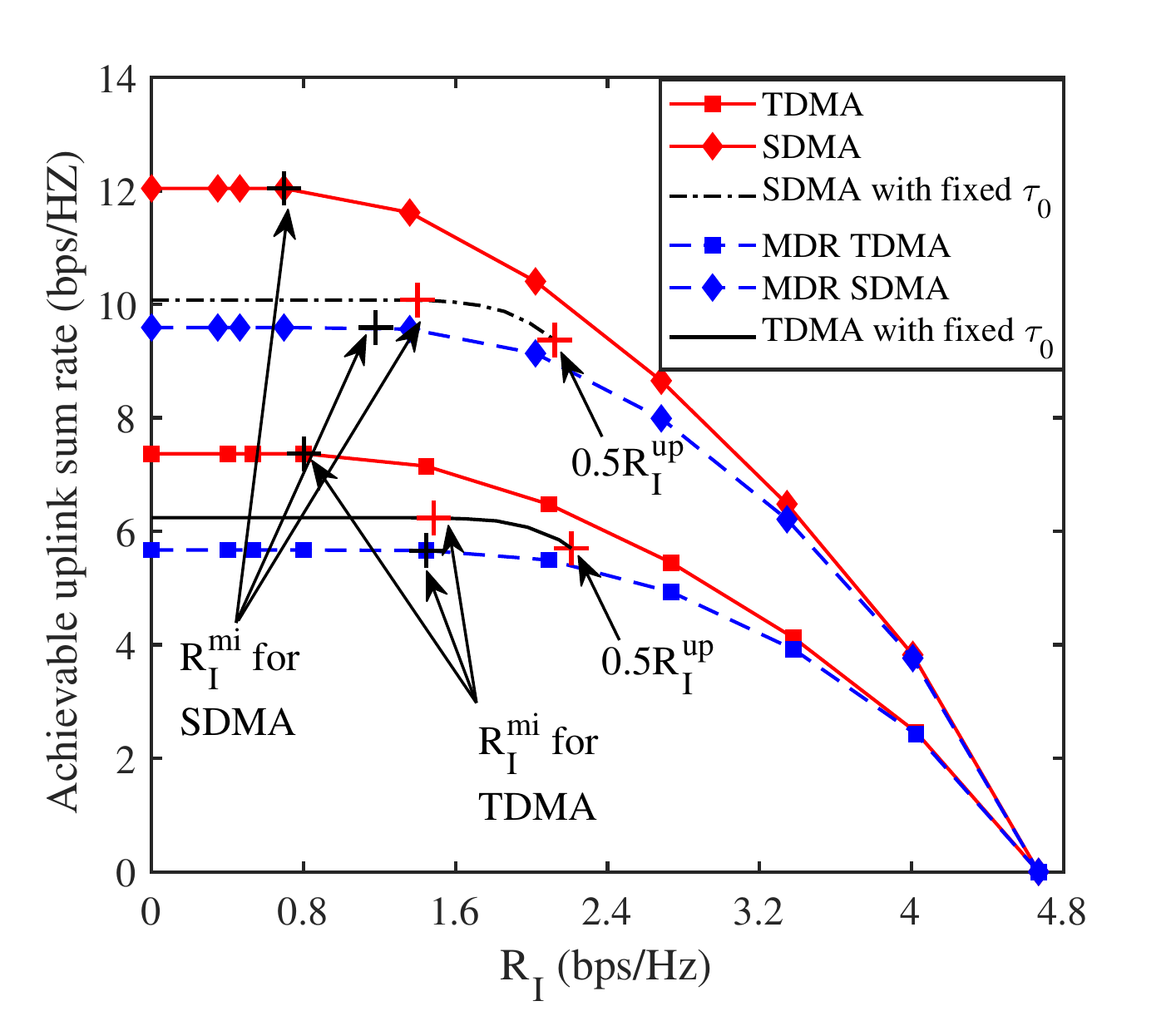}
\caption{Achievable uplink sum rate versus downlink rate threshold $R_I$ for both WPSNs. }
\label{Fig4}
\end{figure}
\begin{figure}[t]
\centering
\includegraphics[width=2.38in]{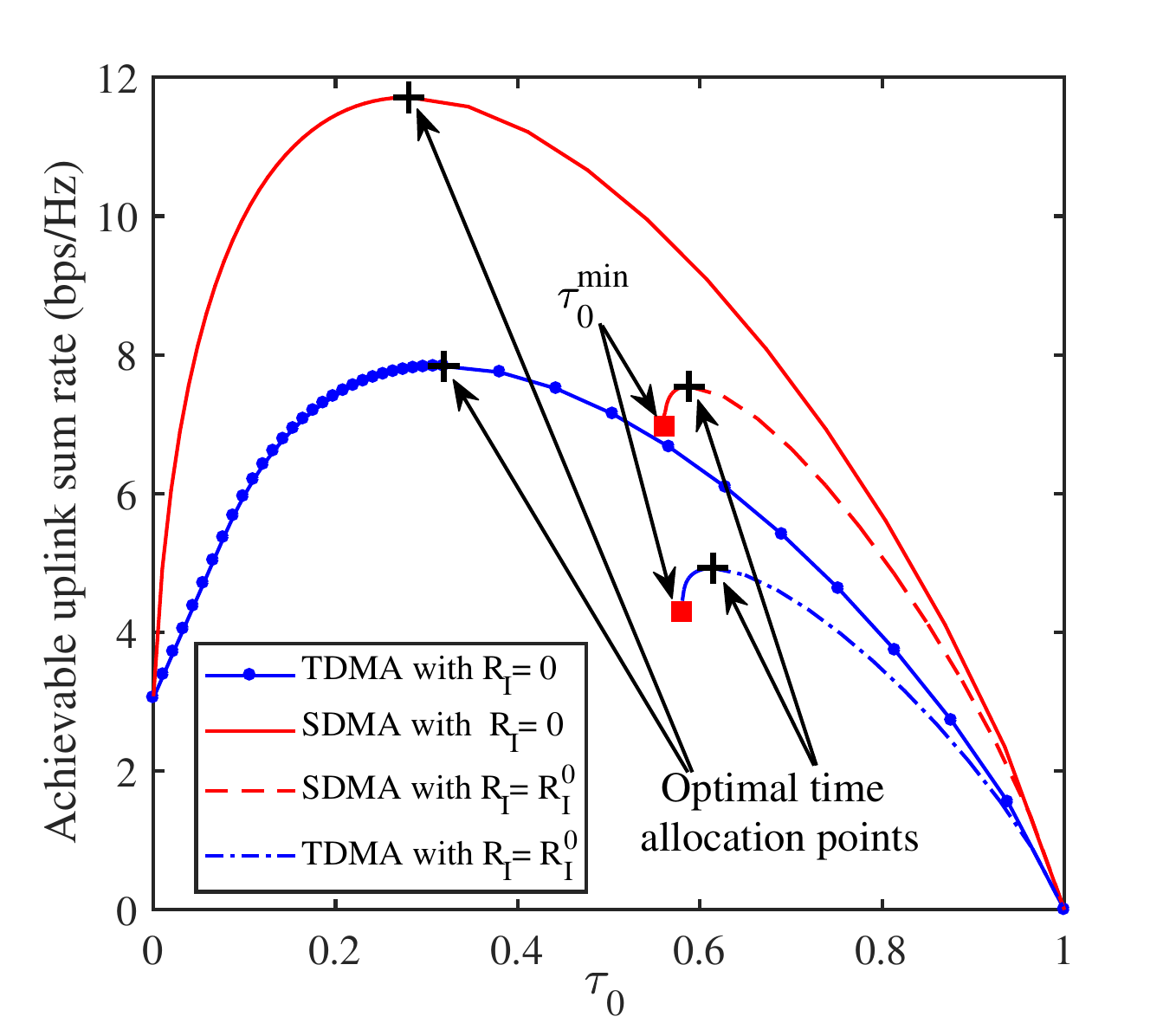}
\caption{ Achievable uplink sum rate versus downlink time duration  $\tau_0$ for both WPSNs, where both $R_I=0$ and $R_I=R_I^0$ are considered.} %
  \label{Fig5}
\end{figure}

 Fig.~\ref{Fig5} finally depicts  the achievable uplink sum  rate
 as a function of the downlink time duration $\tau_0$ for both SDMA-enabled  and TDMA-enabled WPSNs. Two downlink rate thresholds  $R_I\!=\!0$ and $R_I\!=\!R_I^0$ satisfying $ R^{mi}< R_I^0
\le R^{up} $   are  considered, respectively. It is clear from  Fig.~\ref{Fig5} that for both WPSNs, the achievable uplink sum rate $f_S({\tau}_0)/f_T({\tau}_0)$ is indeed  concave  w.r.t.  $\tau_0$ under both thresholds $R_I$. Particularly, when $R_I\!=\!0$, both WPSNs achieve the same uplink sum rate at the point $\tau_0\!=\!0$ since only transmission from the energy-stable IR happens in the uplink.
 However, in the case of  $R_I=R_I^0$,   the minimum downlink time duration $\tau_0^{min}$
is required to satisfy  the  downlink rate constraints of both WPSNs.

\section{Conclusion}
In this paper,  we  investigated the uplink sum rate maximization for both  SDMA-enabled  and TDMA-enabled partially WPSNs.  Different from most existing WPSNs related works, the downlink simultaneous wireless information and power transfer was considered and downlink rate constraint was taken into account in our optimal design. After analyzing the downlink rate constraint and converting the original non-convex uplink sum rate maximization problems into concave ones, semi-closed-form optimal solutions for downlink beamforming, uplink beamforming and time allocation were proposed. Global optimality was proved and low complexity of the proposed optimal solutions were justified. Moreover, from the analysis we found that downlink rate constraint played a significant role and required special care in the optimal design. Finally, numerical simulations verified  the excellent performance of  the  proposed uplink sum rate optimization schemes for both WPSNs.

\vspace{-3mm}
\begin{appendices}
\section{ }
We  can  prove    Lemma~\ref{lemm0}   by  contradiction  as follows.  Firstly, given a downlink rate threshold $R_I^1$ satisfying $R_I^{mi}< R_I^1\le R_I^{up}$, we denote the corresponding maximum objective value of the problem  \eqref{eq6} as $f_{obj,R_I^1}^{\star}(\tau_{0,1}^{\star},\bm{P}_{E_k,1}^{\star},\bm{W}_{B,1}^{\star}, \forall k)$, where $\tau_{0,1}^{\star}$, $\bm{W}_{B,1}^{\star}$ and $\bm{P}_{E_k, 1}^{\star},\forall k$  are the optimal solutions to the problem   \eqref{eq6} with the downlink rate threshold $R_I=R_I^1$.
Meanwhile, we  assume  $\tau_0\log(1\!+
 \!\sigma_n^{-2}\bm{h}_{I_R}^H\bm{W}_{B,1}^{\star}
 \bm{h}_{I_R})>R_I^1$.  Based on this assumption,
it is readily concluded that when another  downlink
 threshold ${R}_I^2$ satisfying  ${R}_I^2=\tau_0\log(1\!+
 \!\sigma_n^{-2}\bm{h}_{I_R}^H\bm{W}_{B,1}^{\star}
 \bm{h}_{I_R})> R_I^{1}$ is applied to the problem  \eqref{eq6},  the previous obtained solution $\bm{\mathcal{Q}}_1\!\!=\!\!\{\tau_{0,1}^{\star},{\bm{W}}_{B,1}^{\star},\bm{P}_{E_k, 1}^{\star},\forall k\}$  actually becomes a  feasible solution for the problem \eqref{eq6} with  $R_I=
{R}_I^2$, since all the constraints are satisfied. So  we  have
 \begin{align}\label{ea1}\vspace{-2mm}
   f_{obj,R_I^1}^{\star}(\bm{\mathcal{Q}}_1) \!\le \! f_{obj,R_I^2}^{\star}(\bm{\mathcal{Q}}_2)
 \end{align} where $\bm{\mathcal{Q}}_2\!\!=\!\!\{\tau_{0,2}^{\star}$,$\bm{W}_{B,2}^{\star}$ $\bm{P}_{E_k, 2}^{\star},\forall k\}$  is the corresponding optimal solution to the problem \eqref{eq6} with ${R}_I={R}_I^2$. On the other hand, since ${R}_I^2> R_I^{1}$, the  set of feasible solutions  of the problem \eqref{eq6} with ${R}_I={R}_I^2$  becomes  smaller than that with ${R}_I={R}_I^1$, thus  we have   \begin{align}\label{ea3}\vspace{-2mm}
   f_{obj,R_I^1}^{\star}(\bm{\mathcal{Q}}_1)\! \ge \! f_{obj,R_I^2}^{\star}( \bm{\mathcal{Q}}_2).
   \end{align}
   By combining  \eqref{ea1} and \eqref{ea3}, it is readily concluded  that
    \begin{align}\label{ea4} \vspace{-2mm} f_{obj,R_I^1}^{\star}(\bm{\mathcal{Q}}_1)=
    f_{obj,R_I^2}^{\star}(\bm{\mathcal{Q}}_2).
   \end{align}

 Similarly, for an arbitrary threshold ${R}_I\in [R_I^{1}, R_I^{2}]$, the same   maximum objective value of the problem \eqref{eq6} can also be obtained, which implies that the downlink rate constraint actually  does not affect the problem \eqref{eq6} and thus can be ignored. As discussed in Section III. B, this happens  only when $ R_I^{2}\le R_I^{mi}$, which contradicts with the original assumption of $R_I^{mi}< R_I^1 < R_I^{2}\le R_I^{up}$. Therefore, the initial assumption is invalid,  we must have the optimal downlink beamforming located at the boundary, i.e., $\tau_0\log(1\!+
 \!\sigma_n^{-2}\bm{h}_{I_R}^H\bm{W}_{B}^{\star}
 \bm{h}_{I_R})=R_I$  for the  problem \eqref{eq6}  when $R_I^{mi}< R_I\le R_I^{up}$. Secondly, since $R_I^{mi}$ can be considered to be a critical point at which the downlink rate constraint becomes tight, we finally conclude that the  downlink rate constraint is tight when $R_I^{mi}\le R_I\le R_I^{up}$.  This completes the proof.
\section{ }
Firstly, we consider the special case of  $R_I=R_I^{up}$ for the  problem \eqref{eq21}, in which  both the constraints $\widetilde{\text{CR}}1$ and   $\widetilde{\text{CR}}2$  are  not strictly feasible (must be tight) by recalling the problem \eqref{eq20}.  In this case,
 only one feasible solution of $\bm{W}_B^{up}=\frac{P_B}{\Vert\bm{h}_{IR}\Vert^2}\bm{h}_{IR}\bm{h}_{IR}^H$  exists for the  problem \eqref{eq21}, so it  is also  globally optimal. Meanwhile, the non-negative dual variables $\lambda^{\star}$ and $\beta^{\star}$ corresponding to constraints
 $\widetilde{\text{CR}}1$ and   $\widetilde{\text{CR}}2$ can be set randomly since they are irrelevant to $\bm{W}_B^{up}$. For the following analysis,  we adopt  $\lambda^{\star}\!=\!0$ and    $ \beta^{\star}\!=\!+\infty$ when  $R_I\!=\!R_I^{up}$. In the sequel, we mainly consider the case of $0\le R_I<R_I^{up}$ in which the problem \eqref{eq21} is strictly feasible and semi-closed optimal
  solutions of $\{
\bm{W}_B^{\star}, {\bm{P}}_{E_k}^{\star}, \forall k\}$ can be found.
Based on  the
definitions $\widetilde{\bm{W}}_B^{\star}\!\!=\!\!
\tau_0{\bm{W}}_B^{\star}$,  $\widetilde{\bm{P}}_{E_k}^{\star}
=(1\!-\!\tau_0) \bm{P}_{E_k}^{\star}, \forall k \!\in \!\mathcal{K}$ and the Lagrangian function in \eqref{eq22},
the Karush-Kuhn-Tucker (KKT)  conditions of the  problem \eqref{eq21} w.r.t.
$\bm{W}_B^{\star}$  and ${\bm{P}}_{E_k}^{\star}$ can be written as
\begin{subequations}
\begin{align}
&\widetilde{\bm{H}}^{\star}+\bm{Z}_0^{\star}=
\bm{0}_{N_B}, \bm{Z}_0^{\star}\bm{W}_B^{\star}=\bm{0},~\bm{Z}_0^{\star}
\succeq \bm{0},\label{eq39}\\
&\bm{G}_{E_k}^H\bm{M}_{\mathcal{K}\setminus k}^{{\star}\frac{1}{2}}( \bm{I}_{N_B}\!\!+\!\sigma_n
^{\!-\!2}\bm{M}_{ \mathcal{K}\setminus k}^{{\star} \frac{1}{2}}\bm{G}_{E_k}{\bm{P}}_{E_k}^{\star}
\bm{G}_{E_k}^H\bm{M}_{ \mathcal{K}\setminus k}^{{\star}\frac{1}{2}}
)^{\!-\!1}\nonumber\\
&~~~~~~~\times \bm{M}_{\mathcal{K}\setminus k}^{{\star} \frac{1}{2}}\bm{G}_{E_k}\!=\!{\sigma_n^{2}}{\ln 2}(\mu_k^{\star}\bm{I}_{N_U}
  \!-\!\bm{Z}_k^{\star}),\label{eq392}\\
  &\bm{Z}_k{\bm{P}}_{E_k}^{\star}=\bm{0},~\bm{Z}_k^{\star}\succeq\bm{0}, \forall k\label{eq393}\\
&\lambda^{\star}(\text{tr}(\bm{W}_B^{\star})-P_B)=0, \label{eq40}\\
&\beta^{\star}(\text{tr}(\bm{h}_{I_R}^H{\bm{W}}_B^{\star}
 \bm{h}_{I_R})-\sigma_n^2(2^{\frac{R_I}{{\tau}_0}}-1))=0, \label{eq41}\\
&\mu_k^{\star}((1\!\!-\!\!\tau_0)\text{tr}({\bm{P}}_{E_k}^{\star})- {\tau}_0 \varepsilon_k\text{tr}
 (\bm{H}_{E_k}\bm{W}_B^{\star}\bm{H}_{E_k}^H))=0, \forall k, \label{eq42}
\end{align}
\end{subequations}
The optimal solutions of $\{\bm{W}_B^{\star}, \lambda^{\star},\beta^{\star}\}$ and $\{{\bm{P}}_{E_k}^{\star}, \forall k, \mu_k^{\star}\}$ can then be derived individually based on above KKT conditions as follows.

\subsection{The optimal $\bm{W}_B^{\star}$, $\lambda^{\star}$ and $\beta^{\star}$}
Recalling that $\widetilde{\bm{H}}^{\star}\!\!=\!\!{\bm{H}}^{\star}\!-\!
\lambda^{\star}\bm{I}_{N_B}$ in \eqref{eq23}, we firstly define  the eigenvalue decomposition (EVD) $ {\bm{H}}^{\star}\!\!=\!\!
  {\bm{U}}_{{\bm{H}}}{\bm{\Lambda
  }}_{{\bm{H}}}{\bm{U}}_{{\bm{H}}}^H,$
  where  the  maximum eigenvalue  $\lambda_{{\bm{H}},\max}$  in the diagonal matrix ${\bm{\Lambda
  }}_{{\bm{H}}}$   satisfies $\lambda_{{\bm{H}},\max}\!\!>\!\!0$
  since  ${\bm{H}}^{\star} \!\succeq\!\bm{0}$ according to its definition in \eqref{eq23}. Then $\widetilde{\bm{H}}^{\star}$ can be rewritten  as
  $\widetilde{\bm{H}}^{\star}\!\!= \!\!{\bm{U}}_{{\bm{H}}}({\bm{\Lambda
  }}_{{\bm{H}}}\!-\!
  \lambda^{\star}\bm{I}_{N_B}){\bm{U}}_{{\bm{H}}}^H$.  Further based on the  KKT  condition  \eqref{eq39}, we simultaneously have
 $\widetilde{\bm{H}}^{\star}=-\bm{Z}_0^{\star}\preceq \bm{0}$ and $\widetilde{\bm{H}}^{\star}\bm{W}_B^{\star}=
\bm{0}$. Therefore,  $\widetilde{\bm{H}}^{\star}$  must be singular and  thus  the optimal  $\lambda^{\star}$  satisfies $\lambda^{\star}\!=\!\lambda_{{\bm{H}},\max}\!>\!0$. Moreover,  we have $\bm{W}_B^{\star}\!=\!c\bm{u}_{{\bm{H}}}\bm{u}_{{\bm{H}}}^H$, where   $\bm{u}_{{\bm{H}}}$ is the unit-norm  eigenvector of ${\bm{H}}^{\star} $  corresponding to $\lambda_{{\bm{H}},\max}$.
 and $c$ is a  constant. Meanwhile, according to the KKT condition \eqref{eq40}, we have $\text{tr}(\bm{W}_B^{\star})=P_B$ due to $\lambda^{\star}\!>\!0$, so  the constant $c\!=\!P_B$ is derived.  As a result, the  optimal $\bm{W}_B^{\star}$ and $\lambda^{\star}$ are finally given  as  \eqref{eq250} in Theorem~\ref{pop1}, i.e.,
 \begin{align}\label{u0}
(\bm{W}_B^{\star},\lambda^{\star}) \! \! =
\!\!\left\{\!\!\!\begin{array}{ll}
  ( P_B\bm{u}_{{H}}\bm{u}_{{H}}^H,\lambda_{{H}}^{\max}) & 0\!\!\le\! \!R_I \!\!<\!\! R_I^{up}\\
 (\frac{P_B}{\Vert\bm{h}_{IR}\Vert^2}\bm{h}_{IR}\bm{h}_{IR}^H, 0) &R_I\!\!=\!\!R_I^{up}
\end{array}\right..
 \end{align}

 As for the optimal $\beta^{\star}$, we firstly observe that when $ 0\le R_I < R_I^{mi}$, as discussed in Section III. B, the downlink rate constraint is inactive, i.e., ${f}_{R}\left(\bm{W}_B^{\star}\right)\!=\!\tau_0\log(1\!+\!\sigma_n^{-2}\bm{h}_{I_R}^H\bm{W}_B^{\star}
 \bm{h}_{I_R})>{ R_{I}}$. It is equivalent to $\text{tr}(\bm{h}_{I_R}^H{\bm{W}}_B^{\star}
 \bm{h}_{I_R})>\sigma_n^2(2^{\frac{R_I}{{\tau}_0}}-1)$. Based on  the KKT condition \eqref{eq41}, it directly yields that $\beta^{\star}\!\!=\!\!0$ when $ 0\le R_I < R_I^{mi}$. On the other hand, when $  R_I^{mi} \!< \!R_I\! < \!R_I^{up}$, recalling Lemma~\ref{lemm0}, the downlink rate constraint is tight, i.e.,  ${f}_{R}\left(\bm{W}_B^{\star}\right)\!=\!
 {\tau}_0\log(1\!+\!\sigma_n^{-2}\bm{h}_{I_R}^H
\bm{W}_B^{\star}
 \bm{h}_{I_R})={ R_{I}}$.  Thus the optimal $\beta^{\star}$ which affects the optimal $\bm{W}_B^{\star}$ implicitly through ${\bm{H}}^{\star}$ should satisfy the nonlinear equation ${f}_{R}\left(\bm{W}_B^{\star}\right)\!=\!{ R_{I}}$. As shown in Appendix A,   $R_I^{mi}$ is a critical point at which the inactive downlink rate
constraint becomes tight, so we can still adopt $\beta^{\star}\!\!=\!\!0$ for $R_I=R_I^{mi}$. Now combining the above results with that for the special case of $R_I=R_I^{up}$ discussed in the beginning, we can summarize the optimal  $\beta^{\star}$ as
 \begin{align} \label{beta}
 \beta^{\star}\!\!= \!\!\left\{\!\!\!\begin{array}{ll}
  0& 0\!\le\! R_I \!\le\! R_I^{mi} \\
{\arg} ~\{{f}_{R}\left(\bm{W}_B^{\star}\right)\!\!=\!\!{ R_{I}}\} &R_I^{mi}\!\!<\! R_I\! \!< \!\!R_I^{up}\\
+\infty & R_I\!\!=\!\!R_I^{up}
\end{array}\right.,
\end{align}
 as shown  in \eqref{eq254}.
\subsection{The optimal $ {\bm{P}}_{E_k}^{\star}$ and $\mu_k^{\star}, \forall k$}
To derive the optimal ${\bm{P}}_{E_k}^{\star}, \forall k$, all related KKT conditions  \eqref{eq392}, \eqref{eq393} and \eqref{eq42} should be jointly considered. It is noted that these KKT conditions  are with the same structure as that of the conventional MIMO rate maximization problem subject to  the total  transmit power  constraint, as shown in \cite{RR2}. Therefore, the optimal structure of ${\bm{P}}_{E_k}^{\star}, \forall k$ derived in \cite{RR2} is also applicable to our problem. Specifically,
based on the SVD ${\bm{M}}_{\mathcal{K}\setminus k}^{ \star \frac{1}{2} }\bm{G}_{E_k} \!\! =\!\!
   \bm{U}_{M_{\mathcal{K}\setminus k}}^{\star}\bm{\Lambda}_{M_{\mathcal{K}\setminus k}}^{\star}
   \bm{V}_{M_{\mathcal{K}\setminus k}}^{\star H}, \forall k$, the optimal ${\bm{P}}_{E_k}^{\star}$ is given by ${\bm{P}}_{E_k}^{\star}=
\bm{V}_{M_{\mathcal{K}\setminus k}}^{\star} {\bm{\Lambda}}_{P_{E_k}}\bm{V}_{M_{\mathcal{K}\setminus k}}^{\star H}, \forall k$,
where the diagonal matrix ${\bm{\Lambda}}_{P_{E_k}}=\text{diag}[\Lambda_{P_{E_k},1},
\cdots,\Lambda_{P_{E_k},N_U}]$  needs to be optimized. Further based on  \eqref{eq392} and \eqref{eq393}, we have
 \begin{align}\label{eq452}
&\big({\sigma_n^{2}}{\ln 2}\mu_k^{\star}\bm{I}_{N_U}\!-\!\bm{G}_{E_k}^H\bm{M}_{\mathcal{K}\setminus k}^{{\star}\frac{1}{2}}\big( \bm{I}_{N_B}\!+\!\sigma_n
^{\!-\!2}\bm{M}_{ \mathcal{K}\setminus k}^{{\star} \frac{1}{2}}\bm{G}_{E_k}{\bm{P}}_{E_k}^{\star}\nonumber\\
&\times
\bm{G}_{E_k}^H\bm{M}_{ \mathcal{K}\setminus k}^{{\star}\frac{1}{2}}
\big)^{\!-\!1} \bm{M}_{\mathcal{K}\setminus k}^{{\star} \frac{1}{2}}\bm{G}_{E_k}
 \big){\bm{P}}_{E_k}^{\star}\!=\!\bm{0}
\end{align}
   Then by  substituting  the analytical structure of  ${\bm{P}}_{E_k}^{\star}$ into  \eqref{eq452},
  the optimal   $\Lambda_{P_{E_k},i}$ can be  derived  as \eqref{eq252}, i.e.,
\begin{align}\label{u1}
  \Lambda_{P_{E_k},i}\!\!=\!\bigg[\frac{1}{\ln 2\mu_k^{\star}}-
\frac{\sigma_n^2}
{\Lambda_{M_{\mathcal{K}\setminus k}^{\star},i}^2}\bigg]^+\!\!
\!,~\forall i,  \forall k.
\end{align}

Next, to   obtain the optimal $\mu_k^{\star},\forall k$, we firstly prove $\mu_k^{\star}\!>\!0,\forall k$  for the problem \eqref{eq21}.   Specifically, based on \eqref{eq392} and \eqref{eq393},   we have \begin{align}\label{u2}
{\sigma_n^{2}}{\ln 2}\bm{Z}_k^{\star}&={\sigma_n^{2}}{\ln 2}\mu_k^{\star}\bm{I}_{N_U}-\bm{G}_{E_k}^H
\bm{M}_{\mathcal{K}\setminus k}^{{\star}\frac{1}{2}}\big( \bm{I}_{N_B}\! \nonumber\\
&~~~+\!\sigma_n
^{\!-\!2}\bm{M}_{ \mathcal{K}\setminus k}^{{\star} \frac{1}{2}}\bm{G}_{E_k}{\bm{P}}_{E_k}^{\star}
\bm{G}_{E_k}^H\bm{M}_{ \mathcal{K}\setminus k}^{{\star}\frac{1}{2}}
\big)^{-1}\bm{M}_{\mathcal{K}\setminus k}^{{\star} \frac{1}{2}}\bm{G}_{E_k}\nonumber\\
&\preceq \sigma_n^{2}{\ln 2}\mu_k^{{\star}}\bm{I}_{N_U}
\end{align}
 It is clear that $\mu_k^{\star}\!\! >\!\!0,\forall k$ must hold since $\bm{Z}_k^{\star}\succeq\bm{0}$ .
Then based on the KKT condition \eqref{eq42},
we  have $(1\!\!-\!\!\tau_0)\text{tr}({\bm{P}}_{E_k}^{\star})\!= \!(1\!-\!{\tau}_0)\text{tr}(\bm{\Lambda}_{P_{E_k}})\!\!=\!\!{\tau}_0 \varepsilon_k\text{tr}
 (\bm{H}_{E_k}\bm{W}_B^{\star}\bm{H}_{E_k}^H)$. By  substituting  \eqref{u1} into this equation, the optimal
   $\mu_k^{\star}\!=\! \frac{N_U(1\!-
\!{ {\tau}_0})}{
\ln 2\left(\!{ {\tau}_0} \varepsilon_k
\text{tr}(\bm{H}_{E_k}\bm{W}_B^{\star}
\bm{H}_{E_k}^H)\!+\!\sum\limits_{i\!=\!1}^{N_U}
\frac{(1\!-\!{ {\tau}_0})\sigma_n^2}
{\Lambda_{M_{{\mathcal{K}\setminus k}},i}^2}\!\right)},\forall k,$
is obtained  as \eqref{eq253}, which finally completes the proof.
\section{ }
Based on the analysis in Section III. A, it is easy to obtain that the achievable downlink rate function ${f}_{R}\left(\bm{W}_B
\right)$ is bounded as $0 \le {f}_{R}\left(\bm{W}_B \right) \le R_I^{up}$.
Now given a downlink rate threshold $R_I$ with $R_I^{mi}\!<\! R_I\!<\!R_I^{up}$, we firstly  consider $\beta
 \!\!=\!\!0$  and prove
 ${f}_{R}\left(\bm{W}_B
\right)\vert_{\beta
 \!=\!0}\!\!<\!\!R_I$ by contradiction as follows. If  ${f}_{R}\left(\bm{W}_B
\right)\vert_{\beta
 \!=\!0}\!\ge\! R_I$ holds,  we readily  observe that
 the optimal $\bm{W}_B^{\star}$  is only  determined  by  the ERs  downlink channels in the form of \eqref{eq250}, and actually not influenced by the downlink rate constraint. However, as analyzed in Section III.~B, this case only happens when the downlink threshold $R_I$ is no more than $R_I^{mi}$, i.e.,  $R_I\le R_I^{mi}$, which clearly  contradicts  with the original assumption of $R_I^{mi}\!<\! R_I\!<\!R_I^{up}$.  Thus  it proves that ${f}_{R}\left(\bm{W}_B
\right)\vert_{\beta
 =0}<R_I$. On the other hand, when $\beta\!\rightarrow\!+\infty$,   we have $\frac{{\bm{H}}^{\star}
 }{\beta}\!\rightarrow \!\sigma_n^{-2}{\bm{h}_{I_R}}{\bm{h}_{I_R}
 }^H$ and  $\bm{u}_{{H}}\!\rightarrow \! \frac{\bm{h}_{I_R}}{\Vert\bm{h}_{I_R}\Vert}$ from Theorem~\ref{pop1}. Together with the bounded property of ${f}_{R}\left(\bm{W}_B
\right)$, it follows that ${f}_R(\bm{W}_B)\vert_{\beta
 =+\infty}
\!= \!{R_I^{up}}$. Then we can naturally conclude that  $0 \le {f}_{R}\left(\bm{W}_B
\right)\vert_{\beta
 =0}\!<R_I<\!{f}_{R}\left(\bm{W}_B
\right)\vert_{\beta
=+\infty}\!=\!{R_I^{up}}$.

Next, given $R_I^{mi}\!<\! R_I\!<\!R_I^{up}$, we  prove that  there  exists a unique $\beta^{\star}\!\in\! (0,\infty)$ satisfying ${f}_{R}\left(\bm{W}_B^{\star}
\right)\!\! =\! \!R_I$  by contradiction as follows. Since the problem \eqref{eq21} is strictly concave, there exists one unique and globally  optimal $\bm{W}_B^{\star}$. Assuming that there are two $\beta_1$ and $\beta_2$ $(\beta_1\! <\! \beta_2)$  simultaneously  realizing   the  optimal $\bm{W}_B^{\star}$, and
 denoting the corresponding composite matrices ${\bm{H}}^{\star}\succeq \bm{0}$ in Theorem~\ref{pop1} as ${\bm{H}}_1^{\star}\succeq \bm{0}$ and ${\bm{H}}_2^{\star}\succeq \bm{0}$,  respectively, then we readily have
  \begin{align}\label{ae1}
  {\bm{H}}_1^{\star}\! = \!{\bm{H}}_2^{\star}\!+\!\bm{H}_{I_R}, ~~\bm{H}_{I_R}\!=\!{(\beta_2\!\!-\!\!\beta_1)\bm{h}_{I_R}\bm{h}_{I_R}^H}
  .
  \end{align}
Recalling \eqref{eq250},  it is clear  that  for  obtaining the  unique $\bm{W}_B^{\star}\!=\!P_B\bm{u}_{{H}}\bm{u}_{{H}}^H$,   the dominant eigenspaces of  ${\bm{H}}_1^{\star}$ and ${\bm{H}}_2^{\star}$  should be identical.  Therefore, we can define the EVDs ${\bm{H}}_i^{\star}\!\!=\!\!
  \bm{U}_{H_i}\bm{\Lambda}_{H_i}\bm{U}_{H_i}^H$, where the unitary matrix is $\bm{U}_{H_i}\!\!=\!\![\bm{u}_H,\widetilde{\bm{U}}_{H_i}]$ and the diagonal matrix
  is $\bm{\Lambda}_{H_i}\!=\!\text{diag}[\lambda_{H_i,\max},
  \lambda_{H_i,2},\cdots, \lambda_{H_i,N_B}]\succeq \bm{0}$. Based on $\bm{u}_H^H\widetilde{\bm{U}}_{H_1}\!=\!\bm{0}$ and $\bm{u}_H^H\widetilde{\bm{U}}_{H_2}\!=\!\bm{0}$, we have $\widetilde{\bm{U}}_{H_1}\!=\!\widetilde{\bm{U}}_{H_2}
  \bm{Q}$, where $\bm{Q}\!\in\!\mathbb{C}^{(N_B\!-\!1) \times (N_B\!-\!1)}$ is an arbitrary unitary matrix.  Further by defining $\widetilde{\bm{\Lambda}}
  _{H_i}\!=\!\text{diag}[
  \lambda_{H_i,2},\cdots, \lambda_{H_i,N_B}],~i=1,2$, the equation  \eqref{ae1} can be  rewritten  as
    \begin{align}\label{ae2}
&{\bm{H}_{I_R}}={\bm{H}}_1^{\star}-{\bm{H}}_2^{\star}\\
&\!= \!\underbrace{(\lambda_{H_1,\max} \! \!- \! \!\lambda_{H_2,\max})}_{\bigtriangleup\lambda}\bm{u}_H\bm{u}_H^H \!+ \!\widetilde{\bm{U}}_{H_2}    \underbrace{(\bm{Q}\widetilde{\bm{\Lambda}}
  _{H_1}\bm{Q}^H \!\!- \!\!\widetilde{\bm{\Lambda}}
  _{H_2})}_{\widetilde{\bm{Q}}
  _{H}}\widetilde{\bm{U}}_{H_2}^H\nonumber
  \end{align}
Since  $\bm{u}_H^H\widetilde{\bm{U}}_{H_2}\!=\!\bm{0}$  denoting the orthogonal space, to guarantee the  rank-1 positive semidefinite ${\bm{H}_{I_R}}$, only the following two cases are possible:
 \subsubsection{case1: $\bigtriangleup\lambda\!\neq\! 0$ and $\widetilde{\bm{U}}_{H_2}\widetilde{\bm{Q}}
  _{H}^{\frac{1}{2}}\!\!=\!\!\bm{0}$}
 In this case,  we readily  have  ${\bm{H}_{I_R}}={\bigtriangleup\lambda}\bm{u}_H\bm{u}_H^H $ from   \eqref{ae2} and
 thus the  optimal downlink beamforming $\bm{W}_B^{\star}\!=\!P_B\bm{u}_{{H}}\bm{u}_{{H}}^H\!=\!
 P_B\bm{h}_{I_R}\bm{h}_{I_R}/\Vert\bm{h}_{I_R}\Vert^2$ is derived, which implies that  ${f}_{R}\left(\bm{W}_B^{\star}
\right)\!=\!R_I^{up}\! > \! R_I$.
 \subsubsection{case2: $\bigtriangleup\lambda\!=\!0$ and $\text{rank}(\widetilde{\bm{U}}_{H_2}\widetilde{\bm{Q}}
  _{H}^{\frac{1}{2}})\!=\!1$}
 In this case, ${\bm{H}_{I_R}}\!=\!\widetilde{\bm{U}}_{H_2}    (\bm{Q}\widetilde{\bm{\Lambda}}
  _{H_1}\bm{Q}^H \!\!- \!\!\widetilde{\bm{\Lambda}}
  _{H_2})\widetilde{\bm{U}}_{H_2}^H$  is observed from \eqref{ae2} and thus we have
  $\bm{u}_{{H}}^H{\bm{H}_{I_R}}\bm{u}_{{H}}=0$, which implies that   ${f}_{R}\left(\bm{W}_B^{\star}
\right)\!=0\! < \! R_I$.

It is obvious that
 both the two cases contradict  with the original  equation of ${f}_{R}\left(\bm{W}_B^{\star}
\right)\!=\! R_I$.
 As a result, we   conclude that there is only one  optimal $\beta^{\star}\!\in\! (0,\infty)$ satisfying ${f}_{R}\left(\bm{W}_B^{\star}
\right)\!=\! R_I$.
Meanwhile, ${f}_R(\bm{W}_B)$ is also  continuous  due to the  EVD operation on   positive semidefinite
matrices  and  $\log$ function. Overall, when $R_I^{mi}\!<\! $$R_I\!<\!R_I^{up}$, with the boundness, uniqueness and  continuity of ${f}_R(\bm{W}_B)$,  we easily  conclude that  ${f}_{R}\left(\bm{W}_B
\right)$ is  monotonic  w.r.t. $\beta\!\in\! (0,\infty)$. Moreover, since ${f}_{R}\left(\bm{W}_B
\right)\vert_{\beta
 =0}\!<\!{f}_{R}\left(\bm{W}_B
\right)\vert_{\beta
=+\infty}$, we can infer that ${f}_{R}\left(\bm{W}_B
\right)$ is   monotonically increasing
 w.r.t. $\beta\!\in\! (0,\infty)$ and converges to $R_I^{up}$. This completes the proof.
\section{ }
For the problem \eqref{eq27}, the KKT conditions are given by
 \begin{subequations}
\begin{align}
&\eqref{eq39}, \eqref{eq40},  \eqref{eq41} \\
&\frac{\sigma_n^{\!-\!2}}{\ln \!2}\bm{G}_{E_k}^H\!( \bm{I}_{N_B}\!\!\!+\!\!\frac{\sigma_n^{\!-\!2}}{\tau_{E_K}
}\bm{G}_{E_k}\!\widetilde{\bm{P}}_{E_k}^{\star}
\!\bm{G}_{E_k}^H
)^{\!-\!1}\!\bm{G}_{E_k}\!\!\!=\!\!\mu_k^{\star}\bm{I}_{N_U}
  \!\!-\!\!\bm{Z}_k^{\star},\label{eqt1}\\
  &\bm{Z}_k^{\star}\widetilde{\bm{P}}_{E_k}^{\star}=\bm{0},~
  \bm{Z}_k^{\star}\succeq\bm{0}, \forall k\label{eqt2}\\
&\mu_k^{\star}(\text{tr}(\widetilde{\bm{P}}_{E_k}^{\star})\! - \! {\tau}_0 \varepsilon_k\text{tr}
 (\bm{H}_{E_k}\bm{W}_B^{\star}\bm{H}_{E_k}^H))\! \!=\!\! 0, \forall k \label{eqt3}\\
 &\log\det(\bm{I}_{N_B}\!+\!  \frac{\sigma_n^{-2}}{\tau_{E_k}^{\star}} \bm{G}_{E_k}
   \widetilde{\bm{P}}_{E_k}^{\star}\bm{G}_{E_k}^H)-
   \frac{\sigma_n^{-2}}{\ln 2 \tau_{E_k}^{\star}}\text{tr}((\bm{I}_{N_B}\nonumber\\
   &~~~~\!+\!  \frac{\sigma_n^{-2}}{\tau_{E_k}^{\star}} \bm{G}_{E_k}
   \widetilde{\bm{P}}_{E_k}^{\star}\bm{G}_{E_k}^H)^{-1}\bm{G}_{E_k}
   \widetilde{\bm{P}}_{E_k}^{\star}\bm{G}_{E_k}^H)=0\label{eqt4}\\
 &\tau_{I_R}^{\star}(C_R
\! -\!\gamma^{\star})\!=\!0,~\forall k,~
 \tau_{I_R}^{\star}\!+\!\sum\limits_{k=1}^{K}
 \tau_{E_k}^{\star}\!=\!1-\tau_0.\label{eqt5}
\end{align}
\end{subequations}
\subsection{The optimal $\bm{W}_B^{\star}$, $\lambda^{\star}$ and $\beta^{\star}$ }
  Since the TDMA-enabled WPSNs have the same downlink transmission as the SDMA-enabled WPCNs,  it is clear that the $\{\bm{W}_B^{\star}, \lambda^{\star}, \beta^{\star}\}$-related KKT conditions of the problem \eqref{eq27} are exactly the same as that for the problem \eqref{eq21}, i.e., \eqref{eq39}, \eqref{eq40}, and \eqref{eq41}.  So following the same approach in Appendix  B. A,  we can derive the same optimal $\bm{W}_B^{\star}$ and  $\lambda^{\star},\beta^{\star}$ as in \eqref{u0} and \eqref{beta}, respectively,   to the problem \eqref{eq27}.
\subsection{The optimal $ {\bm{P}}_{E_k}^{\star}$ and $\mu_k^ {\star}, \forall k$ }

 The $\{ {\bm{P}}_{E_k}^{\star},\mu_k^ {\star}\}$-related KKT conditions (\eqref{eqt1},\eqref{eqt2}),\eqref{eqt3}) are similar to that of the MIMO capacity  maximization in \cite{RR2}. Therefore, the optimal $ \widetilde{\bm{P}}_{E_k}^{\star}, \forall k$  can be   derived as $
 \widetilde{\bm{P}}_{E_k}^{\star}\!\!=\!\!
\bm{V}_{G_{E_k}}{\bm{\Lambda}}_{P_{E_k}}\bm{V}_{G_{E_k}}^H, \forall k$,
 where  $\bm{V}_{G_{E_k}}$  comes  from the SVD $\bm{G}_{E_k}\!\!=\!\!\bm{U}_{G_{E_k}}\bm{\Lambda}_{G_{E_k}}
 \bm{V}_{G_{E_k}}^H$  and the  positive diagonal   ${\bm{\Lambda}}_{P_{E_k}}\!\!=\!\!\text{diag}[{\Lambda}
   _{P_{E_k},1}, $ $
\cdots,{\Lambda}_{P_{E_k},N_U}]$ needs to be determined. Similarly to that in  Appendix B. B, we can obtain    ${\Lambda}
   _{P_{E_k},i}\!=\![\frac{\tau_{E_k}^{\star}}{\ln2
\mu_k^{\star}}-
\frac{\sigma_n^2\tau_{E_k}^{\star}}{{\Lambda}_{G_{E_k}, i}^2}]^+,~ \forall i$,   by  substituting the analytical structure of $ \widetilde{\bm{P}}_{E_k}^{\star}, \forall k$  into the KKT conditions \eqref{eqt1}  and \eqref{eqt2}.  Further  based on   $\widetilde{\bm{P}}_{E_k}^{\star}
\!\!=\!\!\tau_{E_k}^{\star} \bm{P}_{E_k}^{\star}, \forall k$, we  readily have the optimal  ${\bm{P}}_{E_k}^{\star}\!\!\!=\!\!
\bm{V}_{G_{E_k}} \!\! \bm{\Lambda}_{P_{E_k}}
\!\!\bm{V}_{G_{E_k}}^H/{\tau_{E_k}^{\star}}$ as \eqref{eq2900} in Theorem~\ref{pop2}. Moreover,   we  also  have  $\mu_k^{\star}\! >\! 0, \forall k$ for the problem   \eqref{eq27}  by jointly considering \eqref{eqt1},
 $\bm{Z}_k^{\star}\succeq\bm{0}$ and $\bm{G}_{E_k}^H\!( \bm{I}_{N_B}\!\!\!+\!\!\frac{\sigma_n^{\!-\!2}}{\tau_{E_K}
}\bm{G}_{E_k}\!\widetilde{\bm{P}}_{E_k}^{\star}
\!\bm{G}_{E_k}^H
)^{\!-\!1}\!\bm{G}_{E_k}\succeq\bm{0}$. Then  based on the  KKT condition  \eqref{eqt3}, $\text{tr}(\widetilde{\bm{P}}_{E_k}^{\star})\!=\!\text{tr}({\bm{\Lambda}}_{P_{E_k}})\!=\! {\tau}_0 \varepsilon_k\text{tr}
 (\bm{H}_{E_k}\bm{W}_B^{\star}\bm{H}_{E_k}^H)$ holds. Hence, the optimal
   $\mu_k^{\star}= \frac{N_U\tau_{E_k}^{\star}}{
\ln 2\left({ {\tau}_0} \varepsilon_k
\text{tr}(\bm{H}_{E_k}\bm{W}_B^{\star}
\bm{H}_{E_k}^H)+\sum\limits_{i=1}^{N_U} \frac{\sigma_n^2\tau_{E_k}^{\star}}
{\Lambda_{G_{E_k}\!,i}^2}\right)}, \forall k,$
 is derived  as \eqref{eq291}.
\subsection{The optimal $\bm{\tau}_{up}^{\star}$ and $\gamma^{\star}$  }
  By substituting  the above  analytical structure of  $\widetilde{\bm{P}}_{E_k}^{\star}$ into the $\{\tau_{E_k}^{\star}, \forall k\}$-related  KKT condition \eqref{eqt4}, we have
\begin{align}\label{eqt8}
  {g}_T(\tau_{E_k}^{\star})&\!=\!\sum\limits_{
i\!=\!1}^{N_U} \bigg(\!\log\big(1\!+\!
 \!\frac{
 \sigma_n^{\!-\!2}\overline{{\Lambda}}_{G_{E_k},i}}{
 \tau_{E_k}^{\star}}\big)\nonumber\\
 &-\frac{ \sigma_n^{\!-\!2}\overline{{\Lambda}}_{G_{E_k},i}}{
 \ln2(\tau_{E_k}^{\star}\!+\!
 \sigma_n^{\!-\!2}\overline{{\Lambda}}_{G_{E_k},i})}\bigg)
\!=\!\gamma^{\star},~\forall k
\end{align} where  $\overline{{\Lambda}}_{G_{E_k},i}\!\!=\!\!
{{\Lambda}}_{G_{E_k},i}^2\Lambda_{P_{E_k},i}, \forall i$. Since the problem \eqref{eq27} is strictly concave, the Lagrangian function in \eqref{eq28} should be upper bounded for  the optimization variables {$ \{\bm{\tau}_{up},
\widetilde{\bm{W}}_B,
 \widetilde{\bm{P}}_{E_k},\forall k\}$} in ${\mathcal{A}_T}$. Now with respect to the non-negative variable $\tau_{I_R}$, the corresponding term in \eqref{eq28} should be upper bounded, which implies that $C_R\!-\!\gamma^{\star}\!\le\!0$. Then together with \eqref{eqt8} and the KKT condition \eqref{eqt5}, only the following two cases are possible for deriving the optimal $\bm{\tau}_{up}^{\star}$ and $\gamma^{\star}$.
 \subsubsection{case1: $C_R\!-\!\gamma^{\star}\!\!<\!\!0$}
It directly follows from the KKT condition \eqref{eqt5} that $\tau_{I_R}^{\star}\!\!=\!\!0$  and $\sum\limits_{k\!=\!1}^K \!\tau_{E_k}^{\star}\!=\!\!1\!-\!{\tau}_0$. Then the optimal $\tau_{E_k}^{\star}=({\tau}_{E_k})^+,\forall k$ and $\gamma^{\star}$ can be jointly derived from the equations
$g_T({\tau}_{E_k})
 \!\!=\!\!{\gamma}^{\star}, \forall k$ {and} $\sum\limits_{k=1}^{K}
 ({\tau}_{E_k})^+
 \!=\!1\!-\!{\tau}_0$,   as shown in \eqref{eq2916} for the case of $\gamma^{\star}> C_R$.
\subsubsection{case2: $C_R\!-\!\gamma^{\star}\!\!=\!\!0$}
In this case, we have $\gamma^{\star}\!=\!C_R$. Then referring to \eqref{eqt8}, the optimal $\tau_{E_k}^{\star}=({\tau}_{E_k})^{+},\forall k$ can be  obtained from $ {g}_T({\tau}_{E_k})
 \!\!=\!\!C_R, \forall k$. Moreover, the optimal timing allocation for the IR is determined as $
{\tau}_{I_R}^{\star}\!=\!({\tau}_{I_R})^+ $ and $
{\tau}_{I_R}\!=\!1\!-\!{\tau}_0\!-
\!\sum\limits_{k\!=\!1}^{K}({\tau}_{E_k})^{+}$, which finally completes the proof.

\end{appendices}

\end{document}